%% file: main.tex
\documentclass{article}

\usepackage{microtype}
\usepackage{graphicx}
\usepackage{subfigure}
\usepackage{booktabs} 
\usepackage{tabularx}
\usepackage{multirow}
\usepackage{hyperref}
\usepackage{amsmath,bm}
\usepackage{algorithm}
\usepackage{algorithmicx}
\usepackage{algpseudocode}
\usepackage{caption}
\usepackage{setspace}
\usepackage{mathtools} 
\usepackage{graphicx} 
\usepackage{algpseudocode} 
\usepackage{setspace}
\usepackage[T1]{fontenc}
\usepackage[utf8]{inputenc}  
\usepackage{overpic}
\usepackage[table]{xcolor}
\usepackage{soul}
\usepackage{cite}
\usepackage{tcolorbox}
\tcbuselibrary{listings, breakable}
\usepackage{wrapfig}

\usepackage{xparse}

\usepackage{iclr2024_conference,times}

\iclrfinalcopy 

\usepackage{amsmath}
\usepackage{amssymb}
\usepackage{mathtools}
\usepackage{amsthm}

\usepackage[capitalize,noabbrev]{cleveref}

\hypersetup{
    colorlinks = true,
    citecolor = darkblue,
    linkcolor = {black}
}

\theoremstyle{plain}
\newtheorem{theorem}{Theorem}
\newtheorem{theoremsup}{Theorem}

\theoremstyle{definition}

\theoremstyle{remark}

\usepackage[textsize=tiny]{todonotes}

\input{defines}

\title{Diff-Unfolding: A Model-Based Score Learning Framework for Inverse Problems}

\author{%
  Yuanhao~Wang \\
  Washington University in St.~Louis\\
  \texttt{yuanhao@wustl.edu} \\
  \And
  Shirin~Shoushtari \\
  Washington University in St.~Louis\\
  \texttt{s.shirin@wustl.edu} \\
  \And
  Ulugbek~S.~Kamilov \\
  Washington University in St.~Louis\\
  \texttt{kamilov@wustl.edu} \\
}

\begin{document}
	
	\maketitle
	
	
	
	
	
	
\begin{abstract}
Diffusion models are extensively used for modeling image priors for inverse problems. We introduce \emph{Diff-Unfolding}, a principled framework for learning posterior score functions of \emph{conditional diffusion models} by explicitly incorporating the physical measurement operator into a modular network architecture. Diff-Unfolding formulates posterior score learning as the training of an unrolled optimization scheme, where the measurement model is decoupled from the learned image prior. This design allows our method to generalize across inverse problems at inference time by simply replacing the forward operator without retraining. We theoretically justify our unrolling approach by showing that the posterior score can be derived from a composite model-based optimization formulation.  Extensive experiments on image restoration and accelerated MRI show that Diff-Unfolding achieves state-of-the-art performance, improving PSNR by up to 2 dB and reducing LPIPS by $22.7\%$, while being both compact (47M parameters) and efficient (0.72 seconds per $256 \times 256$ image). An optimized C++/LibTorch implementation further reduces inference time to 0.63 seconds, underscoring the practicality of our approach.
\end{abstract}

\input{intro}

\input{background}
\input{methods}

\input{experiments}
\input{conclusion}

\section{Acknowledgment}
This paper is supported by the NSF CAREER awards under grant CCF-2043134.

{
\small

\bibliographystyle{IEEEtran}

\bibliography{reference.bib}

}

\appendix
\input{supp}
\end{document}

%% file: defines.tex
\definecolor{olive}{rgb}{0.5, 0.5, 0.0}
\definecolor{maroon}{rgb}{0.69, 0.19, 0.38}
\definecolor{celestialblue}{rgb}{0.29, 0.59, 0.82}
\definecolor{darkgreen}{rgb}{0.0, 0.6, 0.0}
\definecolor{grey}{rgb}{0.5,0.5,0.5}
\definecolor{darkblue}{rgb}{0.19, 0.19, 0.62}
\definecolor{silver}{rgb}{0.7,0.7,0.7}
\definecolor{darkcyan}{rgb}{0.0, 0.55, 0.55}

\definecolor{burntorange}{HTML}{F4511E}

\newenvironment{algohighlight}{%
  \begin{tcolorbox}[
    colback=yellow!10,
    colframe=yellow!50!black,
    boxrule=0.4pt,
    sharp corners,
    breakable,
    left=2pt,
    right=2pt,
    top=3pt,
    bottom=3pt,
    width=\dimexpr\linewidth\relax,
    boxsep=-2.5pt,
  ]
}{%
  \end{tcolorbox}
}
\newboolean{algohighlightindent}
\setboolean{algohighlightindent}{false} 

\newcommand{\boldtablep}[1]{\textbf{\textcolor{red!80!black}{#1}}}

\newcommand{\boldtable}[1]{\cellcolor{blue!10!white}\textbf{\textcolor{red!80!black}{#1}}}
\newcommand{\undertable}[1]{\underline{\textcolor{green!70!black}{#1}}}

\makeatletter
\@ifundefined{todo}{%
}{%
}
\makeatother
\def\defn{\,\coloneqq\,}
\newcommand{\xb}{\bm{x}}
\newcommand{\sbm}{\bm{s}}
\newcommand{\yb}{\bm{y}}
\newcommand{\nb}{\bm{n}}

\newcommand{\Ib}{\bm{I}}
\newcommand{\Ab}{\bm{A}}
\newcommand{\eb}{\bm{e}}
\newcommand{\wb}{\bm{w}}
\newcommand{\E}{\mathbb{E}}
\newcommand{\R}{\mathbb{R}}

\newcommand{\argmin}{\mathsf{arg\,min}}

\newcommand{\Dsf}{\mathsf{D}}
\newcommand{\Rsf}{\mathsf{R}}

\newcommand{\Gsf}{\mathsf{G}}

\def\argmin{\mathop{\mathsf{arg\,min}}} 

\newcommand{\boldzero}{\mathbf{0}}

\newcommand{\Ncal}{\mathcal{N}}
\newcommand{\Rcal}{\mathcal{R}}
\newcommand{\Fcal}{\mathcal{F}}
\newcommand{\diff}{\mathrm{d}}
\newcommand{\dd}{\boldsymbol{d}}

\newcommand{\thetabm}{{\bm{\theta}}}

%% file: intro.tex
\section{Introduction}
\label{sec:intro}

Many problems in computational imaging, biomedical imaging, and computer vision can be posed as inverse problems, where the goal is to recover a target image from its noisy and incomplete measurement. Inverse problems are typically ill-posed, necessitating the use of prior information to ensure accurate recovery of the target image. Traditional methods for solving inverse problems rely on iterative optimization with handcrafted priors~\cite{burger2018variational, parikh2014proximal,heide2014flexisp}, while more recent methods leverage data-driven priors obtained via deep learning (DL)~\cite{zhang2019residual,zamir2020cycleisp,qian2022rethinking}. 


\begin{figure*}[t]
  \begin{center}
    \begin{overpic}[width=\columnwidth]{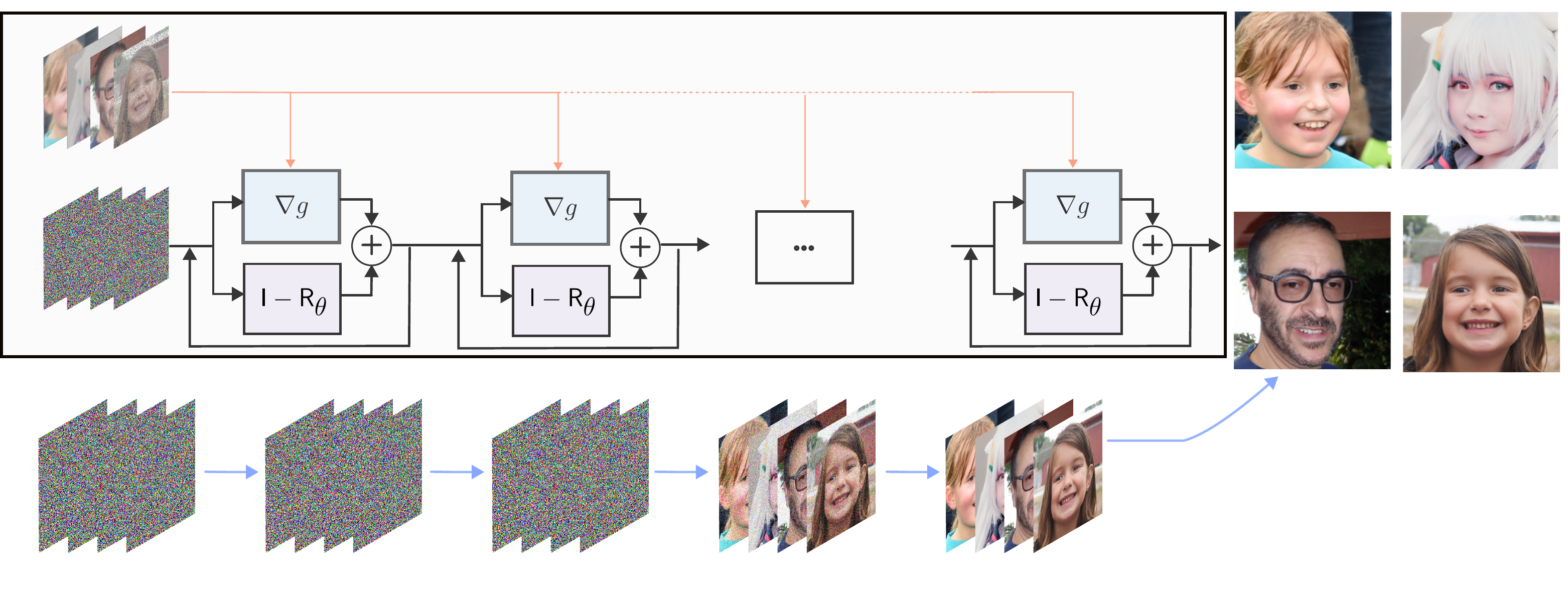}
    \put(6,17){{$\xb_{T}$}}
      \put(6,27){{$\yb$}}
      \put(61, 34) {\small Diff-Unfolding}
      \put(79.2, 25.5) {\tiny Deblur (Gauss)}
    \put(89.9, 25.5) {\tiny Deblur (aniso)}
      \put(80.8, 12.7) {\tiny Inpainting}
      \put(89.5, 12.7) {\tiny Super-resolution}
    \put(15, 32.5){\small \textcolor{burntorange}{Measurements}}
      \put(-0.5,7){$\xb_t$}
    \put(6,0.7){\small $T$}
      \put(20,0.7){\small$T-1$}
      \put(34,0.7){\small$T-2$}
        \put(41,0.7){\small$\cdots$}
      \put(48,0.7){\small$i$}
        \put(55,0.7){\small$\cdots$}
      \put(62,0.7){\small$1$}
        \put(88,0.7){\small$0$}
        
      \put(34,-2.5){\small Sampling Step}
    \end{overpic}
          \vspace{0.1in}
    \caption{\emph{The proposed Diff-Unfolding framework uses a deep unfolding (DU) network as a conditional score function within a diffusion model, enabling explicit separation of the measurement model $\nabla g$ from the learned denoiser $\Rsf_{\thetabm}$ within the architecture. This explicit separation enables direct application of the conditional score function to new inverse problems without retraining.}}
    \label{fig:framework}
  \end{center}
\end{figure*}
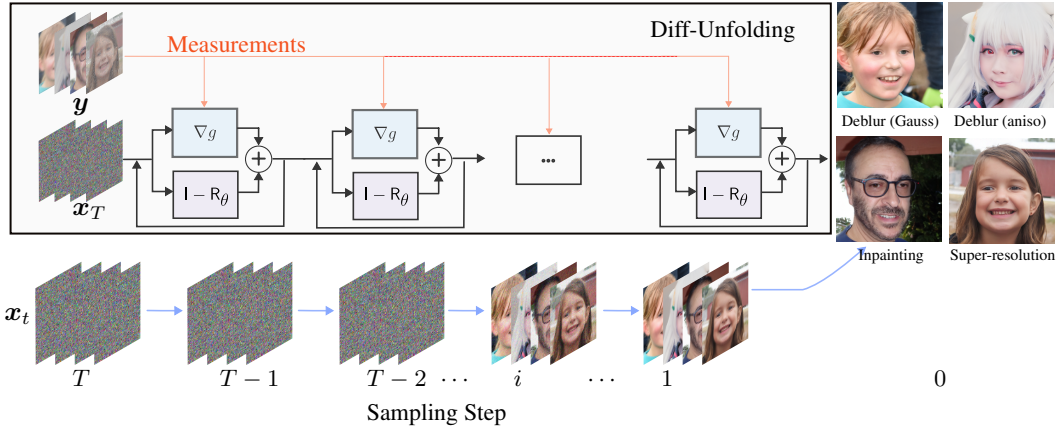

Diffusion models are powerful generative models that achieve state-of-the-art performance in image and video synthesis~\cite{ho2020denoising, song2020denoising, dhariwal2021diffusion,  karras2022elucidating}. These models generate a sample $\xb$ from the target distribution $p(\xb)$ by progressively denoising an initial Gaussian noise $\xb_T$ through the  numerical solution of stochastic differential equations (SDEs) or discrete Markov chains. During training, DMs learn the score functions $\nabla \log p(\xb_t)$, with $t \in \{0, \cdots, T\}$,  of intermediate noisy distributions, enabling approximate sampling from the distribution $p(\xb)$ at $t = 0$. 

The ability of DMs to capture complex image distributions makes them effective as priors for imaging inverse problems (see a recent review~\cite{daras2024survey}). Two main categories of DMs used in this context are \emph{Measurement-Guided Unconditional DMs (UDMs)} and \emph{Conditional DMs (CDMs)}. UDMs use pre-trained unconditional DMs and incorporate measurements only at inference time by approximating the posterior score as $\nabla\log p(\xb_t \,|\, \yb) = \nabla \log p(\yb \,|\, \xb_t) + \nabla \log p(\xb_t)$, where the likelihood term $\nabla \log p(\yb \,|\, \xb_t)$ acts as a guidance signal~\cite{chung2023diffusion, song2023pseudoinverse, zhu2023denoising, wu2024principled}. In contrast, CDMs are trained to directly estimate the conditional score $\nabla\log p(\xb_t \,|\, \yb)$ by incorporating the degraded observation $\yb$ as an input to the model. Conditioning can be implemented through various mechanisms such as input concatenation, cross-attention, or adapter modules~\cite{song2021scorebased, xie2022measurement,LI202247,korkmaz2023self}.

While DM-based inverse solvers show strong empirical performance, they also face technical limitations. A key issue in UDMs is the approximation of the posterior score, which can introduce errors and degrade reconstruction quality. CDMs, on the other hand, are typically specialized for specific inverse problems, limiting their flexibility. Moreover, their reliance on task-specific neural networks reduces interpretability by obscuring how measurements influence the final output. This lack of transparency is particularly concerning given the high computational cost of training DMs.


We present \emph{Diff-Unfolding}, a new general and principled framework learning the posterior score function by systematically integrating the measurement operator into the neural network architecture. Our approach builds on \emph{deep unfolding (DU)}~\cite{ongie2020deep,wen2023physics} to design a modular architecture in which the data-fidelity term is decoupled from the learned prior. This separation enables flexible adaptation to diverse measurement operators without retraining the entire model. We present novel theoretical and numerical results that justify our approach for solving inverse problems. Theoretically, we show that the posterior score can be expressed as a composite optimization problem, with one term capturing the measurement model and the other representing the prior. This interpretation offers a principled foundation for using MBDL to learn the posterior score. Empirically, Diff-Unfolding achieves state-of-the-art results on a range of inverse problems, while being significantly more efficient. Compared to CDMs, it reduces parameter count by $64\%$ and inference time by $50\%$. Compared to DPS, it uses $50\%$ fewer parameters and achieves over $100\times$ faster inference.

%% file: background.tex
\section{Background}
\label{sec:bg}
\textbf{Inverse Problems.} Consider the recovery of the target image $\xb \in \R^n$ from the measurements
\begin{equation}\label{eq:measmodel}
    \yb = \Ab \xb + \eb, 
\end{equation}
where $\Ab \in \R^{m\times n}$ is a measurement operator, and $\eb \in \R^n$ denotes additive white Gaussian noise (AWGN). A commont approach for recovering $\xb$ from $\yb$ is formulating it as an optimization problem
\begin{equation}\label{eq:minim}
	\widehat{\xb} \in \underset{\xb \in \R^n}{\mathsf{arg\,min}} f(\xb) \quad\text{with}\quad f(\xb) = g(\xb) + h(\xb) \quad\text{and}\quad g(\xb) = \frac{1}{2}\|\yb - \Ab\xb\|_2^2, 
\end{equation}
where $g$ is the data-fidelity term that ensures consistency with the measurements and $h$ is the regularization term that encodes prior knowledge on $\xb$. It is common to solve~\eqref{eq:minim} using proximal optimization methods such as ADMM~\cite{boyd2011distributed}, HQS~\cite{he2013half}, and PGD~\cite{beck2009fast}. 


\textbf{Deep Unfolding.} Traditional optimization algorithms require manual parameter tuning, often leading to suboptimal performance. DU is a widely-used MBDL framework that addresses these limitations by reformulating iterative optimization as a trainable neural network, where each iteration becomes a network layer, allowing parameters to be learned in a data-driven fashion~\cite{sun2016deep,yaman2020self,liu2021stochastic,hammernik2018learning}. One widely-used DU architecture combines the gradient descent on $g$ with the residual of a deep network~\cite{aggarwal2018modl,monga2021algorithm,shlezinger2023model}
\begin{equation}
    \label{eq:deep_unfolding}
    \xb^{k+1} = \xb^{k} - \gamma^k \left(\nabla g(\xb^{k}) + \tau^k \left(\xb^{k} - \Rsf_{\thetabm_k}(\xb^{k})\right) \right),
\end{equation}
where $k$ indexes the iteration---interpreted as a  ``layer" in the unfolded network---while $\gamma_k$  and  $\tau_k$  are learnable parameters. $\Rsf_{\theta_j}$ denotes the regularization network, which can either be fixed or learned separately for each iterations. 

\textbf{Diffusion Models.}
DMs are state-of-the-art generative models for sampling from a target distribution~\cite{ho2020denoising, song2021scorebased, dhariwal2021diffusion, karras2022elucidating}. They consist of a forward process that gradually corrupts a data sample $\xb \sim p(\xb)$ by adding Gaussian noise over $T$ steps, following a predefined schedule with standard deviation $\sigma_t$, where $\sigma_0 \approx 0$, $\sigma_T = \sigma_{\mathsf{max}}$, and $\xb_T$ is drawn from a Gaussian distribution $p(\xb_T)$. 

The forward process can be formulated as a Stochastic Differential Equation (SDE)~\cite{song2021scorebased}: $\diff \xb_t = \sqrt{2\dot\sigma_t\sigma_t } \diff \wb_t$, where $\dot\sigma_t$ is the time-derivative of $\sigma_t$ and $\wb_t$ is a standard Wiener process. To generate samples, diffusion models learn to approximate the reverse process by training a neural network $\sbm_{\theta} (\xb_t, \sigma_t)$ to estimate the score function of the intermediate noisy distributions $\sbm_{\theta} (\xb_t, \sigma_t) \approx \nabla_{\xb} \log p_{\sigma_t}( \xb_t )$ as proposed in score matching approaches~\cite{hyvarinen2005estimation}. Once trained, sampling from $p(\xb_0)$ can be performed by solving the reverse-time SDE
\begin{equation}
	\label{eq:ode}
	\diff \xb = -\dot\sigma_t\sigma_t\nabla_{\xb} \log p_{\sigma_t}( \xb_t ) ~\diff t +\sqrt{2\dot\sigma_t\sigma_t } \diff \wb_t. 
\end{equation}
Alternatively, one can also use a deterministic formulation known as the Probability Flow ODE, which shares the same marginals as the SDE but yields a deterministic sampling path. This ODE-based approach enables more efficient numerical solvers while preserving sample quality~\cite{karras2022elucidating, song2021scorebased}.

\textbf{DMs for Inverse Problems.} There are two primary approaches for leveraging DMs to solve inverse problems. Unconditional DMs use Bayes' rule to perform inference-time guidance~\cite{chung2023diffusion, boys2023tweedie, cardoso2023monte, mengdiffusion, rout2024beyond, song2023pseudoinverse, wu2024principled, zhu2023denoising}. These methods typically compute gradients of the log-likelihood $\nabla \log p(\yb \,|\,\xb_t)$, which is often approximated by $\nabla \log p(\yb \,|\,\hat{\xb}_{0|t})$,  where $\hat{\xb}_{0|t}$ is an estimate of the clean image. Conditional DMs (CDMs), in contrast, aim to directly learn the posterior score function $\nabla \log p(\xb_t | \yb)$ by training a neural network to denoise a noisy sample $\xb_t$ conditioned on the observation $\yb$~\cite{song2021scorebased, xie2022measurement, LI202247,korkmaz2023self,choi2021ilvr, saharia2022image}. While often effective, CDMs entangle the measurement model and image prior in a way that complicates generalization to new inverse problems without retraining. Recent work~\cite{korkmaz2023self} introduced SSDiffRecon, an unrolled CDM architecture that combines cross-attention transformers with data-consistency modules for self-supervised MRI reconstruction. Unlike SSDiffRecon, which is tailored to a specific application, our approach builds on a new theoretical framework for posterior score learning, enabling principled and modular architectures applicable across a broad range of inverse problems.

\textbf{Our contributions are:} \textbf{(1)} A new theory showing that the posterior score function can be formulated as a composite optimization problem, with one term representing the physical measurement model and the other representing the image prior. This formulation provides a principled justification for learning the posterior score via deep unfolding. \textbf{(2)} A new modular and principled framework for training posterior score functions by explicitly integrating the measurement model into the network architecture. This design disentangles the roles of data fidelity and the learned prior, enabling efficient adaptation to diverse inverse problems without retraining the entire model. \textbf{(3)} Comprehensive empirical validation showing that Diff-Unfolding achieves state-of-the-art performance across multiple inverse problems while remaining highly efficient. Compared to CDMs, Diff-Unfolding reduces parameter count by up to $64\%$ and inference time by $50\%$. Compared to DPS, it uses $50\%$ fewer parameters and achieves over $100\times$ faster inference.

%% file: methods.tex
\section{Model-based Posterior Score Learning} 
\label{sec:methods}

We now present our Diff-Unfolding framework, which enhance the traditional CDMs methodology using a \emph{deep unfolding} architecture. We start by introducing a theoretical result that justifies our design, showing that the conditional score function can be formulated as a composite optimization problem. We then describe the technical details of the Diff-Unfolding architecture and its implementation.

\subsection{Conditional Score Learning}

In the conventional CDM framework, the task is to sample from the posterior distribution $p(\xb|\yb)$, where $\yb$ is produced by the forward measurement operator in~\eqref{eq:measmodel} for an unknown image $\xb\sim p(\xb)$. The standard forward-diffusion process then corrupts $\xb$ step-by-step by injecting Gaussian noise, yielding a sequence of increasingly degraded images.
\begin{equation}\label{eq:samplingNoisy}
    \xb_t  = \xb + \sigma_t \nb, \quad \nb \sim \Ncal(\bm{0}, \Ib), 
\end{equation}
where $t \in \{0 \cdots, T\}$ represents the diffusion time step and $\sigma_t \in \{\sigma_0, \cdots, \sigma_T\}$ denotes the associated noise level sampled from a predefined noise schedule. The CDM framework integrates the measurement $\yb$ directly into the denoising network $\Dsf_\theta$ aiming to minimize the mean-squared error (MSE) loss function for conditional sampling.
\begin{equation}\label{eq:lossMMSE}
    \E_{\yb \sim p(\yb | \xb), \xb \sim p(\xb), t \sim U[0, T]} \left [ \| \Dsf_\theta (\xb_t, \yb) - \xb\|_2^2\right ],
\end{equation}
where the target image $\xb$ and its measurement $\yb$ form paired samples. The expectation in~\eqref{eq:lossMMSE} is computed over all diffusion steps $t \in [0, T]$, target images $\xb \sim p(\xb)$, and measurements $\yb \sim p(\yb | \xb)$, with conditional likelihood given by the measurement model~\eqref{eq:measmodel}. The minimizer of the MSE loss~\eqref{eq:lossMMSE} is the minimum mean-squared error (MMSE) estimator, which coincides with the conditional expectation of the target given the measurement.
\begin{equation}
    \Dsf_{\theta^\ast}(\xb_t, \yb) = \E[\xb \,|\, \xb_t,\yb].
    \label{Eq:MMSEEstimator}
\end{equation}
Once trained the denoiser $\Dsf_{\theta^\ast}$ can be directly used within the reverse diffusion process to generate samples from the conditional distribution~\cite{LI202247, xie2022measurement,saharia2022image,korkmaz2023self}. 

With the above context established, we now present our main theoretical result.

\begin{theorem}\label{thm:thm1}
    Suppose that the prior density $p(\xb)$ is non-degenerate on $\R^n$. Let $\Dsf_{\theta^\ast}(\xb_t,\yb)$ be the MMSE estimator defined in~\eqref{Eq:MMSEEstimator}.
    Then, the following statements hold for all $t \in \{0, \cdots, T\}$
    \begin{itemize}
        \item[(a)] The denoiser is a minimizer of the following composite optimization problem
        $$\Dsf_{\theta^\ast}( \xb_t,\yb) = \underset{\xb \in \R^n}{\argmin} \left \{ \frac{1}{2} \| \Ab \xb - \yb\|_2^2 +  h_t(\xb)\right \},$$
        for some regularizer $h_t: \R^n \to \R\cup\{+\infty\}$.

        \item[(b)] The denoiser can be related to the following posterior score function 
\begin{equation*}
     \Dsf_{\theta^\ast}(\xb_t,\yb) = \xb_t + \sigma_t^2 \nabla \log p (\xb_t \,|\, \yb),
\end{equation*}
where $p(\xb_t | \yb)$ is the conditional distribution of $\xb_t$ given $\yb$.
    \end{itemize}
\end{theorem}

The proof is provided in Appendix~\ref{app:proof}. Theorem~\ref{thm:thm1} shows that the MMSE estimator in the CDM framework corresponds to the solution of a composite optimization problem with two components: a data-fidelity term encoding the measurement operator and a regularizer $h_t$ capturing the prior knowledge. Thus, computing $\Dsf_{\theta^\ast}$ reduces to minimizing a \emph{regularized least-squares objective}, which can be solved via iterative optimization. 

While the regularizer lacks a closed-form expression, Diff-Unfolding uses deep unfolding to translate iterative optimization steps into a learnable neural architecture, enabling effective inference of its implicit structure. Theorem~\ref{thm:thm1} further establishes a direct connection between the MMSE denoiser—realizable via Diff-Unfolding—and the posterior score function, providing a rigorous justification for its use in the CDM setting.



\subsection{Diff-Unfolding Training and Sampling}

Theorem~\ref{thm:thm1} implies that, given an explicit form for the regularizer $h_t$, the conditional MMSE estimator for CDM can be computed by solving the following composite optimization problem
\begin{equation}
\Dsf_{\theta^\ast}( \xb_t,\yb) = \argmin_{\xb \in \R^n} \left\{ g(\xb) + h_t(\xb) \right\} \quad\text{with}\quad g(\xb) = \frac{1}{2}\|\yb-\Ab\xb\|_2^2.
\end{equation}
We address this optimization problem using deep unfolding, replacing the unknown regularizer $h_t$ with a learnable neural network. A key advantage of this approach is the explicit integration of the measurement model $g(\xb)$ into the architecture, removing the need to learn the measurement operator and improving both interpretability and efficiency.

Algorithms~\ref{alg:train} and~\ref{alg:sampling} outline our training and sampling procedures. At each diffusion step $t \in \{0, \cdots, T\}$ defined by eq.~\eqref{eq:samplingNoisy}, the goal is to learn the MMSE estimator from eq.\eqref{Eq:MMSEEstimator}. Guided Guided by our theoretical framework, we express this estimator as an explicit iterative optimization procedure
\begin{equation}
    \xb_t^{k+1} = \xb_t^{k} - \gamma_t^k \Gsf_{\theta_k}(\xb_t^k)\quad\text{with}\quad \Gsf_{\theta_k}(\xb_t^k) \defn \nabla g(\xb_t^{k}) + \tau_t^k \left(\xb_t^{k} - \Rsf_{\theta_k}(\xb_t^{k})\right),
    \label{Eq:UnfoldingIteration}
\end{equation}
where $g$ is the data-fidelity term and $\Rsf_\theta$ is a learnable neural network representing the regularizer $h_t$. The indices $t \in [0,T]$ and $k \in [1, K]$ in \eqref{Eq:UnfoldingIteration} denote the diffusion step and unfolding iteration, respectively. The parameters $\gamma_t^k$ and $\tau_t^k$ are learnable step-size and regularization parameters for each diffusion step $t$ and iteration $k$. The unfolding iteration is initialized by setting $\xb_t^0 = \xb_t$, corresponding to the noisy image in the diffusion step $t$. After $K$ iterations,the final output  $\xb_t^K$, is evaluated against the target sample $\xb$ via the mean squared error (MSE) loss. As a result, $\Dsf_{\theta^\ast}(\xb_t, \yb)$ approximates the conditional MMSE estimator $\E[\xb \,|\, \xb_t,\yb]$.

\begin{figure*}[t]
\centering
\scriptsize
\begin{minipage}[t]{0.49\textwidth}
\scriptsize
    \makeatletter
    \renewcommand{\alglinenumber}[1]{\small #1}
    \makeatother
    \begin{algorithm}[H]
    \small 
    \caption{Diff-Unfolding Training}
    \begin{algorithmic}[1]
            \State{\bfseries Repeat}
            \State {\quad$\xb \sim p(\xb),\quad \nb \sim  \Ncal(\bm{0}, \bm{I}) $}
            \State{\quad $t \sim  \text{Uniform}([1, T])$}
            \State {\quad$\xb_t = \xb + \sigma_t\nb$}
            \State {\quad $\tau_t^k=f_{\tau}(t,k;\theta_\tau)$, $\lambda_t^k=f_{\lambda}(t,k;\theta_\tau)$}
                \begin{algohighlight}
            \State{\bfseries \quad for {$k = 0,\dots, K-1$}}
                \State{\quad\quad$\xb_t^{k+1} = \xb_t^{k} - \gamma_t^k \Gsf_{\theta_k}(\xb_t^k)$ \textcolor{blue}{\quad (see  eq.~\eqref{Eq:UnfoldingIteration})}}
                \State{\bfseries \quad end for}			
                \end{algohighlight}
            \State{\quad {Take gradient descent on 
            $\mathcal{L}(\theta) = \|\xb_{t}^K -\xb\|_2^2$}}
            \State{\bfseries{until} converged}
    \vspace{0.2em}
    \end{algorithmic}
    \label{alg:train}
    \end{algorithm}
\end{minipage}
\hfill
\begin{minipage}[t]{0.49\textwidth}
    \makeatletter
    \renewcommand{\alglinenumber}[1]{\small #1}
    \makeatother
    \begin{algorithm}[H]
    \small 
    \caption{Diff-Unfolding Sampling}
    \begin{algorithmic}[1]
        \State{\bfseries Input: }{$\yb, \Ab, \{\sigma(t)\}$}
        \State{{\bf init} $\xb_T \sim \mathcal{N} \big( \boldzero, ~\sigma^2(T)  ~\Ib \big)$}
        \State{\bfseries for $t = T, \dots, 0$ do}
        \State {\quad $\tau_t^k=f_{\tau}(t,k;\theta_\tau)$, $\lambda_t^k=f_{\lambda}(t,k;\theta_\tau)$}
        \begin{algohighlight}
            \State{\quad \bfseries for $k = 0,\dots, K-1$ do}
            \State{\quad\quad$\xb_t^{k+1} = \xb_t^{k} - \gamma_t^k \Gsf_{\theta_k}(\xb_t^k)$ \textcolor{blue}{\quad (see  eq.~\eqref{Eq:UnfoldingIteration})}}
            \State{\quad \bfseries end for}
        \end{algohighlight}
        \State{\quad$\xb_{t-1}\sim p(\xb_{t-1}\vert \xb_t, \xb_{t}^{K}, t)$ (see~(81) in~\cite{karras2022elucidating} )}
        \State{\bfseries end for}
        \State{\textbf{return} $\xb_0$}
    \end{algorithmic}
    \label{alg:sampling}
    \end{algorithm}
\end{minipage}
\end{figure*}

We parameterize $\tau_t^k$ and $\lambda_t^k$ using continuous, learnable weighting functions defined as $\tau_t^k=f_{\tau}(t,k;\theta_\tau)$ and $\lambda_t^k=f_{\lambda}(t,k;\theta_\tau)$, where $f_{\tau}$ and  $f_{\lambda}$ are lightweight multilayer perceptrons (MLPs) trained to generate optimal parameters for each diffusion step $t$ and unfolding iteration $k$. As $t$ increases (i.e., at higher noise levels), $\tau_t^k$ increases, placing geater emphasis on the prior. Conversely, for smaller $t$ (lower noise levels), $\tau_t^k$ decreases, prioritizing data consistency and fidelity to the measurements. This adaptive weighting enables the model to transition smoothly from prior-driven denoising in early steps to measurement-driven refinement as noise diminishes. We visualize and further discuss this behaviour in Supplment~\ref{sec:ex_weighting}.

%% file: experiments.tex
\section{Experiments}
\label{sec:exp}


We conduct extensive experiments to evaluate the efficacy, robustness, and computational efficiency of Diff-Unfolding on two key applications: image restoration and accelerated MRI reconstruction under various degradation scenarios. Our assessment includes quantitative metrics---Peak Signal-to-Noise Ratio (PSNR), Structural Similarity Index Measure (SSIM), Learned Perceptual Image Patch Similarity (LPIPS)---alongside visual comparisons to provide a thorough evaluation of our approach.

\subsection{Experiments setup}

\paragraph{Dataset preparation.} We use the FFHQ dataset~\cite{Karras2018ASG} for image restoration. Images indexed 1,000–69,999 are used for training, while 100 images from the first 1,000 images are reserved for testing.  And we evaluate our approach on the fastMRI Multicoil Knee dataset~\cite{Zbontar2018fastMRIAO} for MRI reconstruction. We train on the entire training set (973 scans) and test on three scans from the validation scans. FFHQ images are normalized to the range $\left[0,1\right]$, while fastMRI slices are normalized using their $99th$ percentile absolute value, as recommended in ~\cite{Aali2024AmbientDP}.

\paragraph{Tasks and Degradations.} For image restoration, we evaluate our method on three tasks: \textbf{(1) Deblurring:} using both isotropic and anisotropic Gaussian blur kernel; \textbf{(2) Super-resolution:} with scaling factors of $4\times$, $8\times$ and $16\times$; \textbf{(3) Inpainting:} using random dust-like binary masks patterns with missing probabilities of 0.2, 0.4, and 0.6. For MRI reconstruction, we evaluate three k-space sampling strategies: \textbf{(1) 1D uniform sampling:} $4\times$ and $8\times$ acceleration; \textbf{(2) 1D Gaussian sampling:} $4\times$ and $8\times$ acceleration; and \textbf{(3) 2D Gaussian sampling:} $8\times$ and $16\times$ acceleration. In all experiments, additive i.i.d. Gaussian noise is added to the measurements $\yb \sim \Ncal(0,\sigma_y^2\Ib)$ to simulate realistic measurement conditions.

\paragraph{Training and Testing configuration.} We train Diff-Unfolding using the Adam optimizer~\cite{kinga2015method} with a learning rate of $2\times 10^{-4}$. To promote robustness and generalization, training tasks are sampled randomly across iterations. For the loss function, we adopt an uncertainty-based multi-task learning framework inspired by~\cite{kendall2018multi, karras2024analyzing}, described in detail in~\cref{sec:un_learning}.
For FFHQ images, the measurement noise standard deviation $\sigma_y$ is drawn from the interval $\left[0,0.1\right]$. For fastMRI images, $\sigma_y$ is scaled proportionally to the magnitude of the complex-valued k-space data by a factor of $0.01$. An ablation study on the number of function evaluations (NFE) reveals that the best LPIPS performance is achieved at NFE=18; details are provided in~\cref{sec:abl_nfe}.We evaluated four diffusion models: VP-SDE \cite{ho2020denoising}, VE-SDE \cite{song2021scorebased}, iDDPM \cite{nichol2021improved}, and EDM \cite{karras2022elucidating}, with EDM serving our chosen implementation.

\subsection{Validation on FFHQ }
\begin{figure*}[ht]
	\begin{center}
		\centerline{\includegraphics[width=\columnwidth]{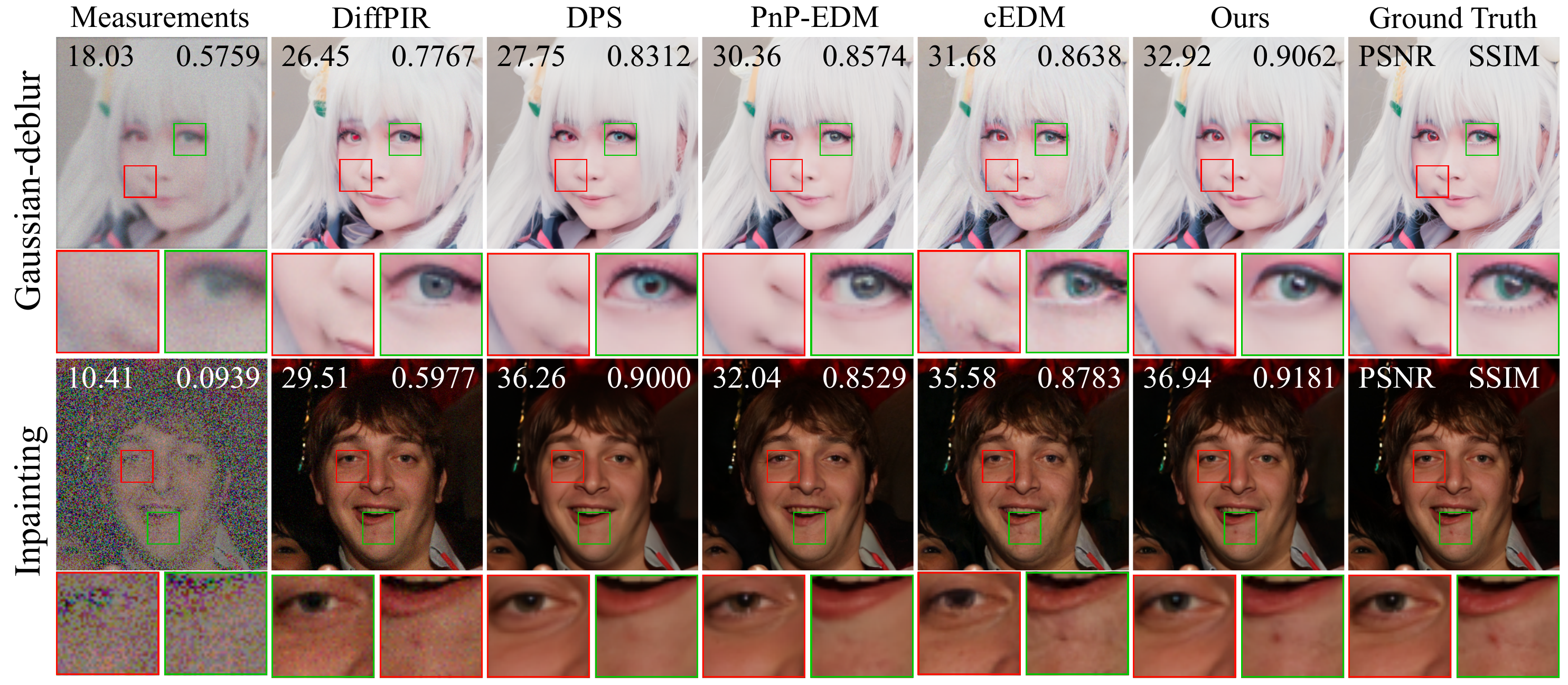}}
		\caption{Visual results for Gaussian deblurring (top) and random inpainting (bottom) with noise level $\sigma_y = 0.05$. Each image is annotated with its corresponding PSNR and SSIM scores. Key visual differences are highlighted under each image. Note the superior quantitative and visual performance achieved by Diff-Unfolding, demonstrating its state-of-the-art performance.}
            \vspace{-5pt}
		\label{fig:img_d_gaussian}
	\end{center}
    \vspace{-5pt}
\end{figure*}

\begin{figure*}[ht]
        \vspace{-5pt}
	\begin{center}
		\centerline{\includegraphics[width=\columnwidth]{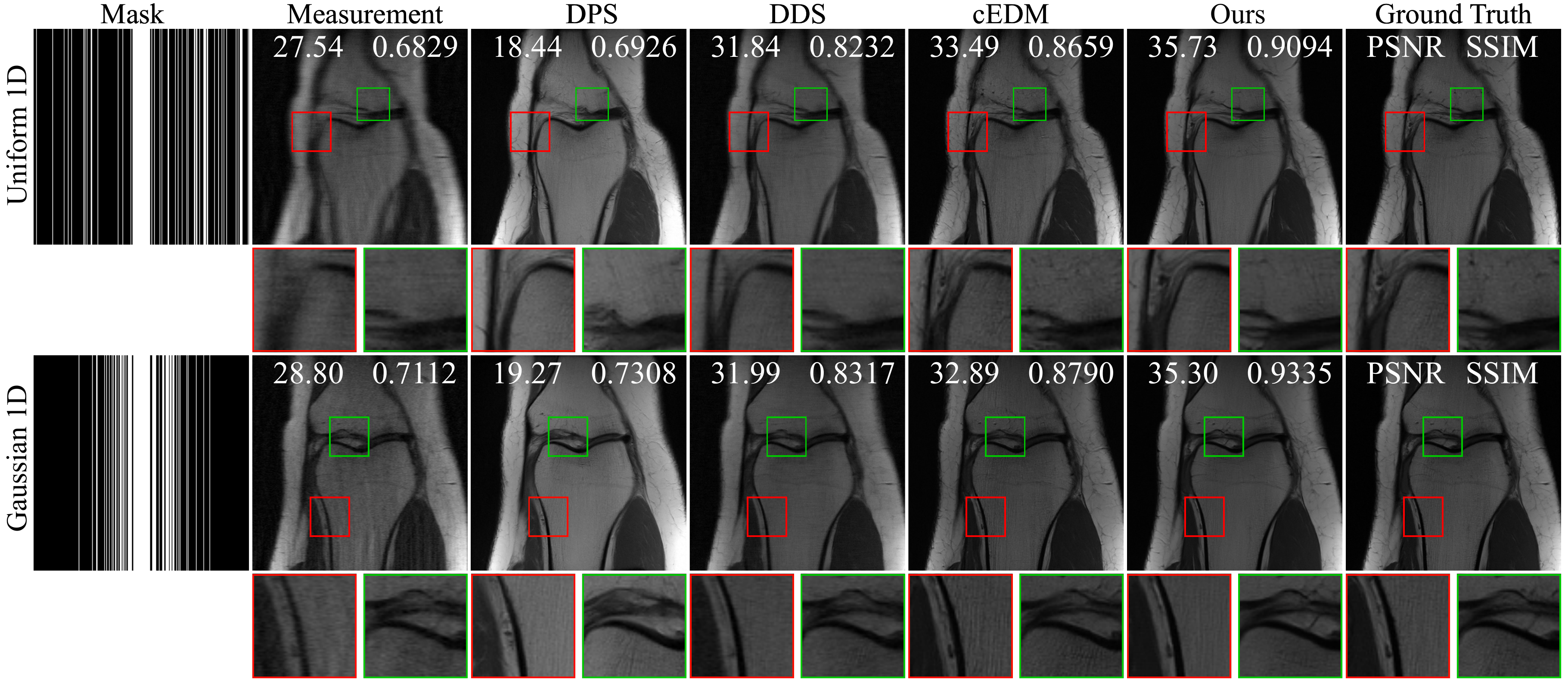}}
		\caption{Visual results for MRI reconstruction using Uniform 1D ($4\times$) and Gaussian 1D ($4\times$) sampling patterns at noise level $\sigma_y = 0.05$. Each image is annotated with PSNR and SSIM values. Key visual differences are highlighted under each image. Diff-Unfolding demonstrates state-of-the-art performance, both quantitatively and visually.}
		\label{fig:mri}
	\end{center}
    \vspace{-5pt}
\end{figure*}

\begin{table*}[!ht]
	\centering
	\setlength{\tabcolsep}{0.5mm}
	\begin{center}
		\scriptsize 
		\caption{Quantitative comparison on three noisy linear inverse problems for 100 FFHQ color test images. \boldtable{Bold:} best; \undertable{undertable:} second best.}
            \vspace{-5pt}

		\begin{minipage}{1\textwidth}
			\centering
                \makebox[\textwidth]{
			\scalebox{1}{ 
				\begin{tabularx}{1\textwidth}{c *{10}{c}} 
					\toprule
					&\multirow{2}{*}{Method} 
					& \multicolumn{3}{c}{Gaussian deblur} 
					& \multicolumn{3}{c}{Super-resolution (4$\times$)} 
					& \multicolumn{3}{c}{Inpainting} \\
					\cmidrule(lr){3-5} \cmidrule(lr){6-8} \cmidrule(lr){9-11}
					&& PSNR ($\uparrow$) & SSIM ($\uparrow$) & LPIPS ($\downarrow$) 
					& PSNR ($\uparrow$) & SSIM ($\uparrow$) & LPIPS ($\downarrow$) 
					& PSNR ($\uparrow$) & SSIM ($\uparrow$) & LPIPS ($\downarrow$) \\ 
					\midrule
					\multirow{7}{*}{\rotatebox{90}{noise = 0.03}} 
					&CDDB &9.92$\pm$0.99 & 0.1520$\pm$0.0705 & 0.7500 & 26.11$\pm$1.57 & 0.6759$\pm$0.0597 &  0.3061 & --- & --- & --- \\
					&DiffPIR &24.32$\pm$ 2.39 & 0.6978$\pm$0.0699 & 0.2806 & 26.48$\pm$ 1.95 & 0.7633$\pm$0.0447 & 0.2391 & 30.47$\pm$1.71 & 0.7910$\pm$0.0234 & 0.2082 \\
					&DPS &26.38$\pm$2.24 & 0.7708$\pm$0.0552 & 0.2324 & 27.58$\pm$2.22 & \undertable{0.8090$\pm$0.0481} & 0.1952 & 33.10$\pm$2.12 & 0.9106$\pm$0.0290 & 0.1562 \\
					&PnP-EDM &29.56$\pm$1.50 & 0.8375$\pm$0.0335 & 0.1854 & 27.46$\pm$1.92 & 0.7873$\pm$0.0476 & 0.2184 & 29.60$\pm$1.47 & 0.8463$\pm$0.0309 & 0.1930 \\
                    &cEDM  & \undertable{31.35$\pm$1.67} & \undertable{0.8678$\pm$0.0293} & \undertable{0.1413} & \undertable{28.22$\pm$2.25} & 0.8049$\pm$0.0462 & \undertable{0.1790} & \undertable{35.10$\pm$1.36} & \undertable{0.9335$\pm$0.0165} & \undertable{0.0953} \\
					\cmidrule(lr){2-11}
					&Ours & \boldtable{{32.50$\pm$1.77}} &\boldtable{0.8977$\pm$0.0243} & \boldtable{0.1294} & \boldtable{29.29$\pm$2.37} & \boldtable{0.8481$\pm$0.0414} & \boldtable{0.1595} & \boldtable{35.48$\pm$1.75} & \boldtable{0.9465$\pm$0.0128} & \boldtable{0.0810} \\


					\midrule
					\multirow{7}{*}{\rotatebox{90}{noise = 0.05}}	
					&CDDB &9.73$\pm$0.98 & 0.1234$\pm$0.0543 & 0.7758 & 23.19$\pm$1.07 & 0.4654$\pm$0.0560 & 0.4747 & --- & --- & --- \\
					&DiffPIR &23.85$\pm$2.32 & 0.6842$\pm$0.0698 & 0.2924 & 25.44$\pm$1.82 & 0.7335$\pm$0.0502 & 0.2639 & 28.09$\pm$1.24 & 0.6737$\pm$0.0351 & 0.2754 \\
					&DPS & 25.96$\pm$1.92 & 0.7570$\pm$0.0554 & 0.2438 & 27.59$\pm$2.20 & \undertable{0.8069$\pm$0.0489} & 0.1986 & 32.14$\pm$1.94 & 0.8976$\pm$0.0295 & 0.1674 \\
					&PnP-EDM & 28.93$\pm$1.57 & 0.8200$\pm$0.0375 & 0.1965 & 27.20$\pm$1.85 & 0.7795$\pm$0.0477 & 0.2260 & 29.56$\pm$1.44 &0.8456$\pm$0.0315  & 0.1931 \\
                    &cEDM & \undertable{29.99$\pm$1.64} & \undertable{0.8336$\pm$0.0322} & \undertable{0.1703} & \undertable{28.04$\pm$2.20} & 0.7984$\pm$0.0462 & \undertable{0.1876} & \undertable{33.81$\pm$1.25} & \undertable{0.9149$\pm$0.0183} & \undertable{0.1209} \\
					\cmidrule(lr){2-11}
					&Ours & \boldtable{31.27$\pm$1.76} & \boldtable{0.8752$\pm$0.0284} & \boldtable{0.1526} & \boldtable{29.16$\pm$2.32} & \boldtable{0.8437$\pm$0.0413} & \boldtable{0.1676} & \boldtable{34.23$\pm$1.59} & \boldtable{0.9306$\pm$0.0158} & \boldtable{0.1037} \\
					\bottomrule
				\end{tabularx}
				\label{tab:img}
			} 
            }
		\end{minipage}
	\end{center}
\end{table*}

\begin{table*}[!ht]
	\centering
	\setlength{\tabcolsep}{0.5mm}
	
	\scriptsize 
	\caption{Quantitative comparison on three noisy linear inverse problems for fastMRI test datasets. \boldtable{Bold:} best; \undertable{Underline:} second best.}
        \vspace{-5pt}
		\begin{minipage}{1\textwidth}
			\centering
                \makebox[\textwidth]{
			\scalebox{1}{ 
			\begin{tabularx}{1.02\textwidth}{c *{10}{c}} 
				\toprule
				&\multirow{2}{*}{  Method} 
				& \multicolumn{3}{c}{Uniform 1D $4
					\times$} 
				& \multicolumn{3}{c}{Gaussian 1D (4$\times$)} 
				& \multicolumn{3}{c}{Gaussian 2D (8 $\times$)} \\
				\cmidrule(lr){3-5} \cmidrule(lr){6-8} \cmidrule(lr){9-11}
				&& PSNR ($\uparrow$) & SSIM ($\uparrow$) & LPIPS ($\downarrow$) 
				& PSNR ($\uparrow$) & SSIM ($\uparrow$) & LPIPS ($\downarrow$) 
				& PSNR ($\uparrow$) & SSIM ($\uparrow$) & LPIPS ($\downarrow$) \\ 
				\midrule
				\multirow{5}{*}{\rotatebox{90}{noise free}} 
				&VarNet & 31.06$\pm$2.66 & \undertable{0.8606$\pm$ 0.0387} & 0.2524 & 33.14$\pm$2.68 & 0.8902$\pm$0.0297 & 0.2151 & 29.51$\pm$2.86 & 0.8480$\pm$0.0339 & 0.2633 \\
				&DPS & 17.92$\pm$5.07 & 0.6977$\pm$0.1607 & 0.2665 & 18.40$\pm$ 5.23 & 0.7489$\pm$0.1521 & 0.2296 & 18.31$\pm$5.00 & 0.7362$\pm$0.1556 & 0.2276 \\
				&DDS & 27.12$\pm$5.29 & 0.8590$\pm$0.0502 & 0.2628 & 31.71$\pm$ 4.41 & \undertable{0.
				9082$\pm$0.0343} & 0.1700 & \undertable{30.08$\pm$4.79} & \undertable{0.8883$\pm$0.0840 }& 0.2253 \\
				&cEDM & \undertable{31.68$\pm$2.17} & 0.8472$\pm$0.0341 & \undertable{0.1927} & \undertable{34.14$\pm$2.77} & 0.8938$\pm$0.0249 & \undertable{0.1503} & 29.40$\pm$2.44 & 0.8444$\pm$0.0306 & \undertable{0.2020} \\
				\cmidrule(lr){2-11}
				&Ours & \boldtable{34.14$\pm$2.72} & \boldtable{0.9058$\pm$0.0220} & \boldtable{0.1616} & \boldtable{37.51$\pm$2.96} & \boldtable{0.9494$\pm$0.0117} & \boldtable{0.0872} & \boldtable{35.90$\pm$3.50} & \boldtable{0.9342$\pm$0.0166} & \boldtable{0.0991}\\
				\midrule
				\multirow{4}{*}{\rotatebox{90}{noise = 0.05}} 
				&DPS &16.88$\pm$4.079 & 0.5899$\pm$0.1423 & 0.3427 & 17.05$\pm$4.25 & 0.6176$\pm$0.1417 &  0.3297 & 17.15$\pm$4.16 & 0.6259$\pm$0.1410 & 0.3274 \\
				&DDS & 26.48$\pm$5.06 & 0.7872$\pm$0.0636 & 0.2807 & 30.01$\pm$4.93 & 0.8417$\pm$0.0456 & 0.2321 & \undertable{29.67$\pm$ 4.53} & \undertable{0.8678$\pm$0.0345} & 0.2418 \\
				&cEDM & \undertable{30.97$\pm$2.55} & \undertable{0.8226$\pm$0.0413} & \undertable{0.2096} & \undertable{32.39$\pm$2.72} & \undertable{0.8510$\pm$0.0317} & \undertable{0.1859} & 28.95$\pm$2.50 & 0.8257$\pm$0.0347 & \undertable{0.2199} \\
				\cmidrule(lr){2-11}
				&Ours & \boldtable{32.21$\pm$2.29}  & \boldtable{0.8658$\pm$0.0304}  & \boldtable{0.1906} & \boldtable{34.11$\pm$2.30} & \boldtable{0.8958$\pm$0.0239} & \boldtable{0.1582} & \boldtable{32.93$\pm$3.82} & \boldtable{0.8766$\pm$0.03288} & \boldtable{0.1792} \\
                \bottomrule
			\end{tabularx}
            }
            }
	\end{minipage}
	\label{tab:mri}
\end{table*}
To evaluate our approach on the FFHQ dataset~\cite{Karras2018ASG}, we compare our approach with the following methods: (1)  
 \textbf{Score approximation methods}: DPS ~\cite{chung2023diffusion}, DiffPIR ~\cite{zhu2023denoising}, and PnP-EDM ~\cite{wu2024principled}.
(2) \textbf{Conditional score estimation}: CDDB ~\cite{chung2023direct} and Conditional EDM (cEDM) ~\cite{karras2022elucidating}.

For explicit approximation methods, we use the pre-trained checkpoint from \cite{chung2023diffusion}, and apply it within the explicit approximation approaches. However, it is important to note that conditional score estimation methods are operator-dependent, meaning each trained model is specified to a particular operator, making the direct comparisons less straightforward.

Notably, CDDB was not trained on noise level $0.05$, resulting in suboptimal Gaussian deblurring performance due to its training on a limited range of noise levels. To address this limitation, we trained cEDM~\cite{karras2022elucidating} by concatenating the measurements to the network's input for each measurement configuration. We conducted comparisons at two noise levels, $0.03$ and $0.05$, with numerical results presented in~\cref{tab:img}. 
Across all three tasks, Diff-Unfolding consistently outperform prior methods in PSNR and SSIM while achieving competitive or superior LPIPS scores. For example, on Gaussian deblur at noise level 0.05, Diff-Unfolding achieves $31.27$ dB PSNR and $0.8752$ SSIM, outperforming cEDM by $+1.28$ dB PSNR, $+0.0416$ SSIM, and $-0.0177$ LPIPS. 
\begin{table}[!ht]
	\centering
	\setlength{\tabcolsep}{0.3mm}
	\begin{center}
		\scriptsize 
		\caption{\small Evaluation of Diff-Unfolding using several sampling methods evaluated on  FFHQ color test images and fastMRI test set with a noise level 0.05. \boldtablep{Bold:} best performance; \undertable{Underline:} second best. Based on these results, we use EDM across all experiments.}
		\begin{minipage}{0.85\textwidth}
			\centering
                \makebox[\textwidth]{
			\scalebox{1}{ 
            \centering
				\begin{tabularx}{\textwidth}{c *{10}{c}} 
					\toprule
					\multirow{2}{*}{Dataset}&\multirow{2}{*}{Method} 
					& \multicolumn{3}{c}{Gaussian deblur} 
					& \multicolumn{3}{c}{Super-resolution (4$\times$)} 
					& \multicolumn{3}{c}{Inpainting} \\
					\cmidrule(lr){3-5} \cmidrule(lr){6-8} \cmidrule(lr){9-11}
					&& PSNR ($\uparrow$) & SSIM ($\uparrow$) & LPIPS ($\downarrow$) 
					& PSNR ($\uparrow$) & SSIM ($\uparrow$) & LPIPS ($\downarrow$) 
					& PSNR ($\uparrow$) & SSIM ($\uparrow$) & LPIPS ($\downarrow$) \\ 
					\midrule
					\multirow{4}{*}{FFHQ}&Ours (VP) & 31.10 & 0.8912 & \undertable{0.2014} & 29.86 & 0.8680 & 0.2245 & 29.96 & 0.9023 & 0.1594 \\
					&Ours (VE) & \undertable{32.03} & \undertable{0.8964} & 0.2157 & \undertable{30.24} & \undertable{0.8696} & 0.2329 & 33.90 & 0.9278 & 0.1783 \\
					&Ours (iDDPM) & \boldtablep{32.14} & \boldtablep{0.9010} & 0.2126 & \boldtablep{30.33} & \boldtablep{0.8746} & 0.2248& \boldtablep{35.01} & \boldtablep{0.9425} & 0.1592  \\
					&Ours (EDM) & 31.27 & 0.8752 & \boldtablep{0.1526} & 29.16 & 0.8437 & \boldtablep{0.1676} & \undertable{34.23} & \undertable{0.9306} & \boldtablep{0.1037} \\\cmidrule(lr){1-11}
                                   \multirow{4}{*}{fastMRI} &Ours (VP) & 30.81 & 0.8341 & 0.2383 & 30.71 & 0.8271 & 0.2034 & 26.21 & 0.7791 & 0.2645 \\
				&Ours (VE) & 31.80 & \undertable{0.8624} & \undertable{0.2217} & 32.46 & 0.8581 & \undertable{0.1901} & \undertable{30.58} & 0.8383 & \undertable{0.2167} \\
				&Ours (iDDPM) & \undertable{31.83} & 0.8523 & 0.2403 & \undertable{33.00} & \undertable{0.8703} & 0.2138 & 29.78 & \undertable{0.8524} & 0.2180 \\
				&Ours (EDM) & \boldtablep{32.21} & \boldtablep{0.8658} & \boldtablep{0.1906} & \boldtablep{34.11} & \boldtablep{0.8958} & \boldtablep{0.1582} & \boldtablep{32.93} & \boldtablep{0.8766} & \boldtablep{0.1792} \\

				\bottomrule
					\bottomrule
				\end{tabularx}
				\label{tab:img_abl}
			} 
            }
		\end{minipage}
	\end{center}
        \vspace{-10pt}
\end{table}

\cref{fig:img_d_gaussian} illustrates visual comparisons for Gaussian deblurring and inpainting. Diff-Unfolding restores fine facial details more accurately and avoids common artifacts like over-smoothing or hallucination seen in DPS and cEDM. Importantly, even under heavy masking or noise, the reconstructions remain stable and visually coherent.

\paragraph{Ablation Study on Different Sampling Variants of DMs.} We conducted an ablation study on four sampling strategies—VP-SDE, VE-SDE, iDDPM, and EDM—using the FFHQ and fastMRI dataset at a noise level of $0.05$. As shown in~\cref{tab:img_abl}, in the FFHQ test, the iDDPM variant achieved the highest PSNR and SSIM, while EDM produced the lowest (best) LPIPS score. However, in the fastMRI test, EDM produced the best result in all tasks. Based on these results, we focus used EDM sampling accross all the experiments.

%
%
\subsection{Validation on fastMRI}

We evaluate our approach on the fastMRI dataset~\cite{Zbontar2018fastMRIAO}. We compare against the following baselines:  \textbf{(1) traditional methods}: E2E-Varnet~\cite{sriram2020end};  \textbf{(2) Score approximation methods}: DPS~\cite{chung2023diffusion} , DDS~\cite{Chung2023DecomposedDS}; and \textbf{(3) Conditional Score esimtaiton}: cEDM~\cite{karras2022elucidating}. More implementation details can be found in~\cref{sec:implementaion}. For the explicit approximation methods, we use the pre-trianed checkpoint from~\cite{Chung2023DecomposedDS}, and apply it within the explicit approximation approaches. 

We consider both noise-free and noisy scenarios. Quantitative results are presented in~\cref{tab:mri}. Across all the settings, Diff-Unfolding outperforms other methods. In the noise-free Uniform 1D ($4\times$), Diff-Unfolding achieves $34.14$ dB PSNR, $0.9058$ SSIM, and $0.1616$ LPIPS, significantly surpassing all baselines by a large margin, $+2.46$, $+0.0586$ and $-0.0311$,substantially outperforming the strongest baseline (cEDM). \cref{fig:mri} shows several visual results. Diff-Unfolding is able to preserve anatomical structures, avoid aliasing artifacts, and reconstruct fine details such as cartilage and soft tissue boundaries. In contrast, DPS and DDS often produce blurry or inconsistent outputs, especially under high undersampling and high noise.

We have also compared Diff-Unfolding against the SSDiffRecon method from~\cite{korkmaz2023self}. However, 
\begin{wraptable}{r}{0.42\textwidth}
\centering
\scriptsize
\vspace{-10pt}
\caption{\small Quantitative comparison with SSDiffRecon~\cite{korkmaz2023self} on three metrics. Our method outperforms SSDiffRecon across all metrics.}
\begin{tabular}{lccc}
\toprule
\textbf{Method} & \textbf{PSNR} $\uparrow$ & \textbf{SSIM} $\uparrow$ & \textbf{LPIPS} $\downarrow$ \\
\midrule
SSDiffRecon & 42.67 & 0.9848 & 0.0538 \\
Ours & \textbf{48.03} & \textbf{0.9903} & \textbf{0.0481} \\
\bottomrule
\end{tabular}
\vspace{-10pt}
\label{tab:ssdfcomparison}
\end{wraptable}
since SSDiffRecon was only implemented to perform self-supervised learning on MRI, we implemented a self-supervised version of Diff-Unfolding adapted to the experimental settings in~\cite{korkmaz2023self}.
We report the corresponding result, which was trained on the same dataset using a 2D Gaussian sampling pattern with an acceleration factor of 4. As can be seen in~\cref{tab:ssdfcomparison},  Diff-Unfolding achieves PSNR of $48.03$, SSIM of $0.9903$, and $0.481$ of LPIPS, outperforming the results by~\cite{korkmaz2023self} (PSNR: 42.67, SSIM: 0.9848, LPIPS: 0.0538).


\subsection{Computational Efficiency}
\label{sec:efficiency_com}
\begin{wraptable}{r}{0.5\textwidth}
	\centering
    \vspace{-30pt}
	\setlength{\tabcolsep}{0.3mm}
	\scriptsize 
	\caption{\small Performance comparison of different methods on FFHQ: Our method achieves faster inference with fewer parameters.}
    \vspace{-5pt}
	\begin{minipage}{0.5\textwidth}
    			\centering
		\begin{tabular}{lcccc}
			\toprule
			\textbf{Dataset} & \textbf{Method} & \textbf{Param. Size(M)}($\downarrow$) & \textbf{NFE} ($\downarrow$) & \textbf{s/image}($\downarrow$) \\
			\midrule
			\multirow{5}{*}{FFHQ} 
			& CDDB & 552.80 & 100 & $108.70$ \\
			& cEDM & 138.28 & 35 & $1.59$ \\
			& DiffPIR & 93.56 & 100 &  $1.52$ \\
			& DPS & 93.56 & 1000 & $78.77$ \\
			& PnP-EDM & 93.56 & 100 & $120.98$ \\
			& Ours & 46.94 & 18  & $0.72$ \\
			\bottomrule
		\end{tabular}
		\label{tab:computation}
	\end{minipage}
    \vspace{-15pt}
\end{wraptable}
We assess the computational efficiency of our model on the FFHQ dataset for the deblurring task using a single RTX Titan GPU. As shown in Table~\cref{tab:computation}, our approach 
achieves the lowest computational cost in terms of both parameter count and inference time, while also attaining the best LPIPS among all baselines. Additionally, we provide a C++/Libtorch implementation of our method, which further reduces inference time to $0.63$ \textbf{s/image}, offering an additioanl performance gain in practical deployment scenarios.

%% file: conclusion.tex
\section{Conclusion}
\label{sec:conclusion}

In this work, we introduced Diff-Unfolding, a novel framework that integrates physics-based deep unfolding with conditional diffusion models to address a broad class of inverse problems. By explicitly modeling the posterior score function during training, Diff-Unfolding aims to approximate the posterior and enables a single model to generalize across diverse forward operators. Extensive experiments on image restoration and accelerated MRI tasks demonstrate that Diff-Unfolding achieves state-of-the-art performance across both quantitative metrics (PSNR, SSIM) and perceptual quality (LPIPS), while using fewer parameters and offering significantly faster inference than existing methods.

\section*{Limitations}

Despite its advantages,Diff-Unfolding has several limitations. While it supports multiple forward operators, the current implementations are restricted to linear degradations and additive Gaussian noise. Future work could extend the framework to handle more complex and nonlinear degradations. Additionally, although Diff-Unfolding is more efficient than previous diffusion-based methods, it currently requires fine-tuning to handle out-of-distribution operators, which is not considered during the training.

\section*{Impact Statement}

By unifying diverse imaging tasks under a single framework Diff-Unfolding could have the potential to reduce the need for task-specific models and training, and lower the computation of inference. Its efficiency and flexibility make it promising for practical deployment in real-world applications, including medical imaging. We expect this work could make a positive impact on computational imaging and inverse problems approaches.

%% file: supp.tex
	\appendix
    \newpage
\section{Proof of Theorem~\ref{thm:thm1}}\label{app:proof}

\begin{theoremsup}\label{thm:thm1}
    Suppose that the prior density $p(\xb)$ is non-degenerate on $\R^n$, and let $\Dsf_{\theta^\ast}(\xb_t,\yb)$ be the MMSE denoiser defined in~\eqref{Eq:MMSEEstimator}.
    Then, the following statements hold for all $t \in [0, T]$
    \begin{itemize}
        \item[(a)] The denoiser satisfies the following composite optimization problem
        $$\Dsf_{\theta^\ast}( \xb_t,\yb) = \underset{\xb \in \R^n}{\argmin} \left \{ \frac{1}{2} \| \Ab \xb - \yb\|_2^2 +  h_t(\xb)\right \},$$
        for some regularizer $h_t: \R^n \to \R\cup\{+\infty\}$.

        \item[(b)] The denoiser can be related to the following posterior score function 
\begin{equation*}
     \Dsf_{\theta^\ast}(\xb_t,\yb) = \xb_t + \sigma_t^2 \nabla \log p (\xb_t \,|\, \yb),
\end{equation*}
where $p(\xb_t | \yb)$ is the conditional distribution of $\xb_t$ given $\yb$.
    \end{itemize}
\end{theoremsup}

\begin{proof}
\noindent
\begin{itemize}
\item [(a)] Since $\Dsf_{\theta^\ast}(\xb_t,\yb)$ is trained with MSE loss
\begin{equation}
    \theta^\ast =\underset{\theta}{\mathsf{arg\,min}} \left\{ \E \left [ \| \Dsf_\theta (\xb_t, \yb) - \xb\|_2^2\right ]\right\},
\end{equation}
it corresponds to the conditional mean estimator $ \E[\xb\mid \xb_t,\yb] $. Expanding expectation, we obtain
\begin{align}
   \nonumber \Dsf_{\theta^\ast}( \xb_t,\yb) &=  \E[\xb \,|\, \xb_t,\yb] = \int \xb~p(\xb | \xb_t, \yb)~d \xb \\
      \nonumber& = \int \xb \frac{p(\yb | \xb_t, \xb )p(\xb, \xb_t)}{p(\xb_t, \yb) }~d \xb \\
    & =  \int \xb \frac{p(\yb | \xb)p(\xb|\xb_t) }{p(\yb|\xb_t) }~\dd \xb  = \E_{\xb\sim p(\xb|\xb_t)}[\xb|\yb], 
\end{align}
where  the second and third lines follow from  Bayes' rule. In the third equality, we also used the assumption  that $\xb_t$ and $\yb$ are conditionally  independent given $\xb$, i.e., \(p(\yb | \xb_t, \xb) = p(\yb | \xb)\). 

$\E_{\xb\sim p(\xb|\xb_t)}[\xb|\yb]$  is the posterior expectation over \(p(\xb | \xb_t)\), updated by the likelihood \(p(\yb | \xb)\). That is, the clean signal \(\xb\) is first predicted based on the diffusion sample \(\xb_t\), and then reweighted by how likely that \(\xb\) would produce the measurement \(\yb\). The last expression implies that we can recondition the prior \(p(\xb | \xb_t)\) using the measurement \(\yb\), and then taking the expectation under this updated posterior.

According to Theorem 2 of~\cite{priorGribonval}, the conditional MMSE estimator \(\mathbb{E}[\xb | \yb]\), where \(\xb \sim p_{\xb}\), minimizes the functional  $\frac{1}{2}\|\Ab\xb - \yb\| + h(\xb)$, where  $h$ is a regularizer derived from the prior $p_{\xb}$. In our setting, the prior is replaced by the conditional distribution \(p(\xb | \xb_t)\), and the corresponding MMSE estimator becomes

\begin{equation}
    \Dsf_{\theta^\ast}(\yb, \xb_t) = \E_{\xb\sim p(\xb|\xb_t)}[\xb|\yb].
\end{equation}
This change in the prior implies that the associated regularizer must also adapt accordingly—from a static regularizer \(h(\xb)\) to a conditional regularizer \(h(\xb | \xb_t)\), which reflects the updated distribution informed by \(\xb_t\).
Putting everything together, we conclude that the estimator can be expressed as the solution to the following optimization problem:
\begin{equation}
    \E_{\xb\sim p(\xb|\xb_t)}[\xb|\yb] = \Dsf_{\theta^\ast}( \xb_t,\yb) = \underset{\xb \in \R^n}{\argmin} \left \{ \frac{1}{2} \| \Ab \xb - \yb\|_2^2 +  h(\xb| \xb_t)\right \}. 
\end{equation}

\item [(b)] We now show that the conditional MMSE estimator \(\Dsf_{\theta^\ast}(\xb_t, \yb) = \mathbb{E}[\xb | \xb_t, \yb]\) can be written in terms of posterior score function \(\nabla_{\xb_t} \log p(\xb_t | \yb)\). Starting from Bayes’ rule, we have 
\begin{equation}\label{eq:bayes_rule_app}
    p(\xb|\xb_t, \yb) = \frac{p(\xb, \xb_t, \yb)}{p(\xb_t, \yb)} = \frac{p(\xb_t |\xb, \yb)}{p(\xb_t, \yb)} = \frac{p(\xb_t |\xb) p(\xb| \yb)}{p(\xb_t| \yb)}
\end{equation}
Let us consider the marginal density 
\begin{equation}
    p(\xb_t|\yb) = \int p(\xb_t|\xb) p(\xb|\yb)~d\xb. 
\end{equation}
Differentiating this with respect to \(\xb_t\), we obtain
\begin{align}
    \nonumber\nabla_{\xb_t} p(\xb_t|\yb) &= \int \nabla_{\xb_t} p(\xb_t|\xb) p(\xb|\yb)~d\xb \\
    \nonumber& = \frac{1}{\sigma_t^2}\int (\xb - \xb_t) p(\xb_t|\xb) p(\xb|\yb)~d\xb\\
    \nonumber& = \frac{1}{\sigma_t^2}\int (\xb - \xb_t) p(\xb|\xb_t, \yb)  p(\xb_t|\yb)~d\xb \\
    & = \frac{p(\xb_t|\yb)}{\sigma_t^2}\int (\xb - \xb_t) p(\xb|\xb_t, \yb)  ~d\xb, 
\end{align}
where in the third line, we used the result from~\eqref{eq:bayes_rule_app}. Rearranging the terms yields 
\begin{equation}
     \E[\xb - \xb_t|\xb_t, \yb] = \sigma_t^2 \frac{\nabla_{\xb_t} p(\xb_t|\yb)}{p(\xb_t|\yb)} = \sigma_t^2 \nabla_{\xb_t} \log p(\xb_t|\yb), 
\end{equation}
and hence
\begin{equation}
     \E[\xb|\xb_t, \yb] = \xb_t +  \sigma_t^2 \nabla_{\xb_t} \log p(\xb_t|\yb).
\end{equation}
\end{itemize}
\end{proof} \normalfont
\section{Operation Setting}
\label{sec:svd_op}
In this work, we employ several key operators for image restoration tasks, including super-resolution, deblurring, and inpainting. For super-resolution and deblurring, we utilize the singular value decomposition (SVD) as suggested in prior works~\cite{kawar2021snips,kawar2022denoising}. SVD is a matrix factorization technique that decompose a matrix to three components: $U$,$\Sigma$, and $V^T$, where $U$ and $V$ are orthogonal matrices, and $\Sigma$ is a diagonal matrix containing the singular values. This decomposition enables effective handling the ill-posed nature of super-resolution and deblurring by regularizing the solution space and preserving important structure information in  images. 

For the inpainting task, we employ random dust-like mask patterns with varying probabilities of 0.2, 0.4, and 0.6. These masks simulate realistic scenarios where portions of the image are missing or corrupted. The inpainting process aims to reconstruct the missing regions by leveraging the surrounding pixel information and prior knowledge of the image structure. This approach is particularly useful for handling occlusions or artifacts in real-world images.

By integrating these operators into a unified training process, we aim to comprehensively address the challenges of image restoration, ensuring robustness and high-quality results across various degradation scenarios with a single training cycle.

\section{Implementation details}
\label{sec:implementaion}

To ensure a fair comparison across experiments, we conduct our experiments on 2 NVIDIA A100-SXM4-80GB GPUS, using Python 3.9.20, PyTorch 2.1, CUDA 11.8, and CuDNN 8.9.4. Model testing is performed on a single RTX Titan GPU, under identical software and environment conditions. Additionally, we further implement our model on C++/LibTorch framework, resulting in a $14\%$ performance improvement in time consumption. For experiments involving MRI reconstruction, we simulate coil sensitivity maps for 15 coils from the \texttt{reconstruction\_rss} of the multi-coil dataset.

\subsection{Implementation details of comparisons}
\paragraph{E2E-VarNet~\cite{hammernik2018learning}} We trained the supervised learning-based method using Uniform 1D, Gaussian 1D, and Gaussian 2D subsampling patterns, aligning with the sampling methods utilized in the paper, and adhering closely to the official implementation in the paper.

\paragraph{SSDiffRecon~\cite{korkmaz2023self}} SSDiffRecon was originally a self-supervised approach implemented in Tensorflow, designed for a 5-coil setup at a resolution of $512\times 512$. We adapted our approach similarly to a self-supervised approach, employing the identical Guassian 2D sub-sampling pattern for both training and testing, ensuring a fair comparison.

\paragraph{cEDM~\cite{karras2022elucidating}} For conditional EMD(cEDM), we implicitly concatenated the measurement data directly to the input of the network, adhering closely to the original EDM settings. For each experiment, we trained a new model.

\paragraph{PnP-EDM~\cite{wu2024principled}}
We followed the PnP-EDM implementation, employing a 50-step EDM sampler backbone, modifying only the forward operator to align with the methodology described in \cref{sec:svd_op}. The pretrained unconditional diffusion checkpoint was taken from~\cite{chung2023diffusion}.

\paragraph{DPS~\cite{chung2023diffusion}} We replicated the original DPS methodology, using a 1000-step DDPM sampler backbone. Similar to PnP-EDM, we modified only the forward operator as detailed in~\cref{sec:svd_op}, utilizing pretrained unconditional diffusion checkpoints from~\cite{chung2023diffusion, Chung2023DecomposedDS}.

\paragraph{DDS~\cite{Chung2023DecomposedDS}} Following the original DDS implementation, we employ a 50-step DDIM sampler backbone without further modifications. The pretrained checkpoint was obtained from ~\cite{Chung2023DecomposedDS}.

\paragraph{CDDB~\cite{chung2023direct}} We directly utilized the original checkpoint provided by the authors without modifications.
\paragraph{DiffPIR~\cite{zhu2023denoising}}. We applied the sample modification as DPS, using the unconditional diffusion model sourced from~\cite{chung2023diffusion}.
\subsection{Explicit Weighting}
\label{sec:ex_weighting}

In conventional deep unfolding (DU) approaches, the parameters $\tau$ and $\lambda$ are typically predefined as learnable constants. In our approach, we accounts for the varying contributions of the data-fidelity and regularization terms across both timesteps and unfolding iterations. Specifically, we model $\tau_t^k = f_{\tau}(t, k; \theta_\tau)$ and $\lambda_t^k = f_{\lambda}(t, k; \theta_\lambda)$ using a one-layer MLP, enabling the parameters to adapt consistently over the sampling timeline. The resulting learned functions are illustrated in~\cref{fig:functau}.

\begin{figure}[ht]
    \centering
    \begin{minipage}[t]{0.75\textwidth}
	\begin{overpic}[width=\columnwidth]{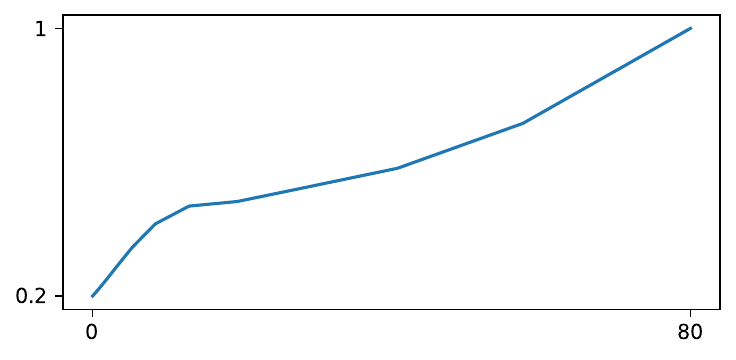}
        \put(52,3){\large $t$}
        \put(1,25){\large $\tau_t$}
    \label{fig:functau}
	\end{overpic}        
        \caption{Learned weight function $\tau_t^k=f_{\tau}(t,k;\theta_\tau)$ when $k=0$. The function varies significantly over time, dynamically adapting to the sampling schedule. At larger timesteps, it assigns greater weight to the diffusion prior, whereas at smaller timesteps, it emphasizes data fidelity.} 
    \end{minipage}\hspace{0.1in}
    \begin{minipage}[t]{0.75\textwidth}

    	\begin{overpic}[width=\columnwidth]{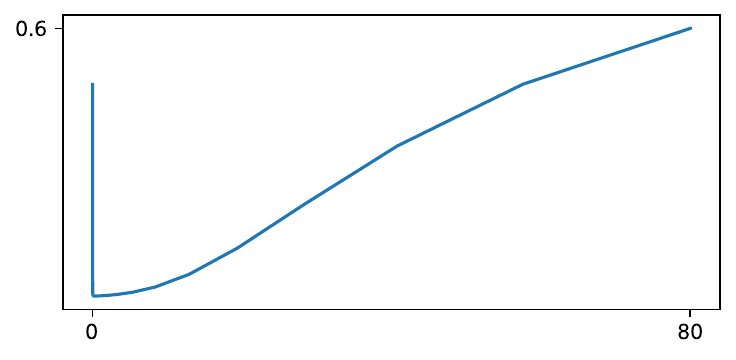}
            \put(52,3){\large $t$}
        \put(1,25){\large $u_t$}
			\label{fig:funcu}
	\end{overpic}   
    \caption{Learned weighted function $u_t = f_u(t;\theta_u)$. The function demonstrates significant changes with different time steps.}
    \end{minipage}
\end{figure}
\subsection{Uncertainty-based Multi-task Learning}
\label{sec:un_learning}
Traditional multi-task learning often aggregates individual losses as a weighted sum,  $\mathcal{L}=\sum_i w_i\mathcal{L}_i$, where $i$ indicates the individual tasks or loss terms within a training batch. However, this method is highly sensitive to the choice of weight $w_i$, particularly when the relative value between different loss terms changes significantly. To address this limitation, we adopt the approach proposed in \cite{kendall2018multi,karras2024analyzing}, and introduce a continuous, learnable weight function that could balance the loss term adaptively.

\begin{equation}
	\mathcal{L}({\theta}, u_t) =\E_\sigma \left[\frac{1}{\exp^{u_t}} \mathcal{L}({\theta};t) +u_t \right]
\end{equation}
Here $u_t = f_u(t;\theta_u)$ is a continuous uncertainty function, and parameterized by a simple one-layer MLP. By modeling log variance, this approach automatically balances the loss terms without manual tuning of weights. 
The optimized $u_t$ is illustrated in ~\cref{sec:un_learning}. This uncertainty-based weighting improves the stability of training and the overall quality of restored images by adaptively focusing on the most reliable information available at each noise level.

\section{Sampling}

Following the EDM~\cite{karras2022elucidating}, our implementation could be simplified as $\Tilde{\yb} =\Rcal_\theta(\xb;\sigma)$, where $\xb$ is a noisy input and $\sigma$ is the corresponding noise level. $\Rcal$ could break down to:

\begin{equation}
    \Rcal_\theta(\xb;\sigma) = c_{skip}(\sigma)\xb + c_{out}(\sigma)\Fcal_\sigma(c_{in}(\sigma)\xb; c_{noise}(\sigma))
\end{equation}
where the inputs and output of the network $\Fcal_{\theta}$ are preconditioned according to $c_{in}$, $c_{out}$, $c_{skip}$ and $c_{noise}$. 
We follow the EDM sampling strategy of the precondition setting~\cite{karras2022elucidating}. The noise schedule, parameters, and preconditioning coefficients vary depending on the chosen noise model. Different settings could be found in ~\cite{karras2022elucidating} Table 1.

\section{Additional experimental results}

\subsection{Ablation study on neural function evaluations}
\label{sec:abl_nfe}
We conducted an ablation study to assess the impact of different neural function evaluations (NFEs) on LPIPS for tasks including Gaussian Deblurring, Super-resolution ($\times 4$), and Inpainting. Our findings suggest that selecting NFE=$18$ provides optimal and consistent performance.
\begin{figure*}[!ht]
	\begin{center}
		\centerline{\includegraphics[width=0.8\columnwidth]{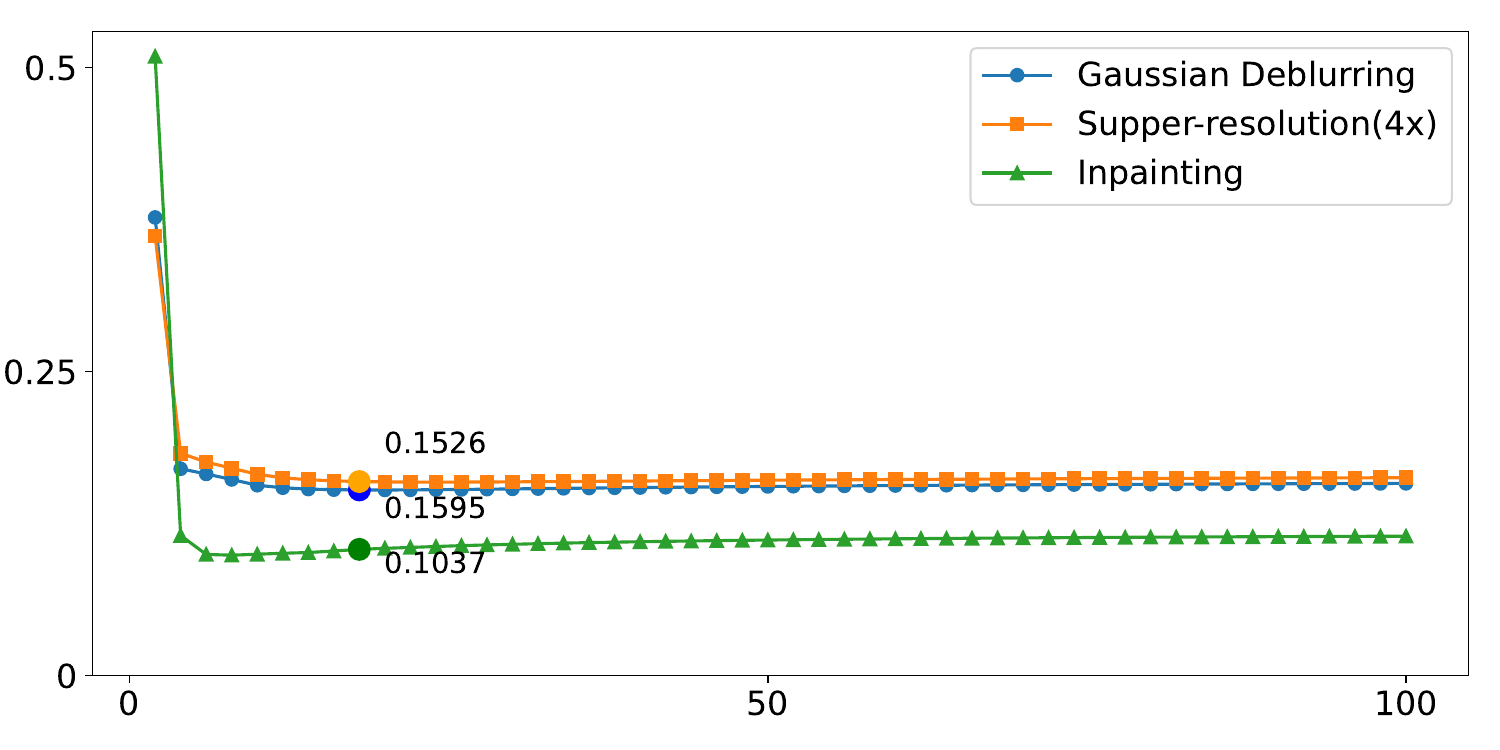}}
		\caption{LPIPS variation with different NFEs.}
		\label{fig:abl}
	\end{center}
\end{figure*}
\subsection{Additional Visual results}
Additional experimental results demonstrate the effectiveness of our approach across various tasks under a noise level of $\sigma_y = 0.05$, as illustrated in~\cref{fig:vgdeblur,fig:vagdeblur,fig:sr8,fig:inpaint40} for FFHQ test dataset and ~\cref{fig:mrig1d8,fig:mrig2d15,fig:mriu1d8} for MRI test dataset.

\newpage
\begin{figure*}[!ht]
	\begin{center}
		\centerline{\includegraphics[width=\columnwidth]{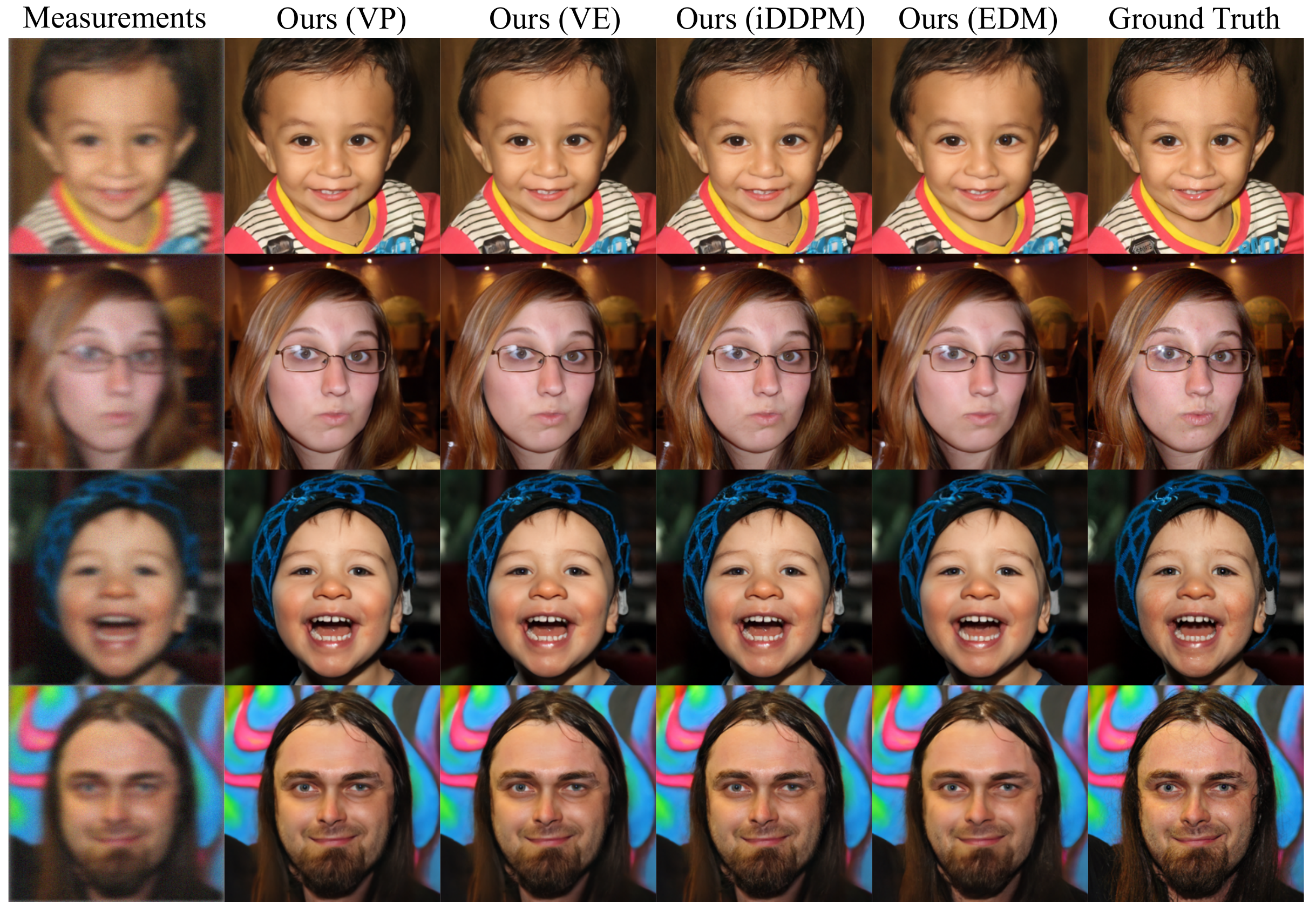}}
		\caption{Visual results of Gaussian deblur}
			\label{fig:vgdeblur}
		\end{center}
\end{figure*}
\begin{figure*}[!ht]
	\begin{center}
		\centerline{\includegraphics[width=\columnwidth]{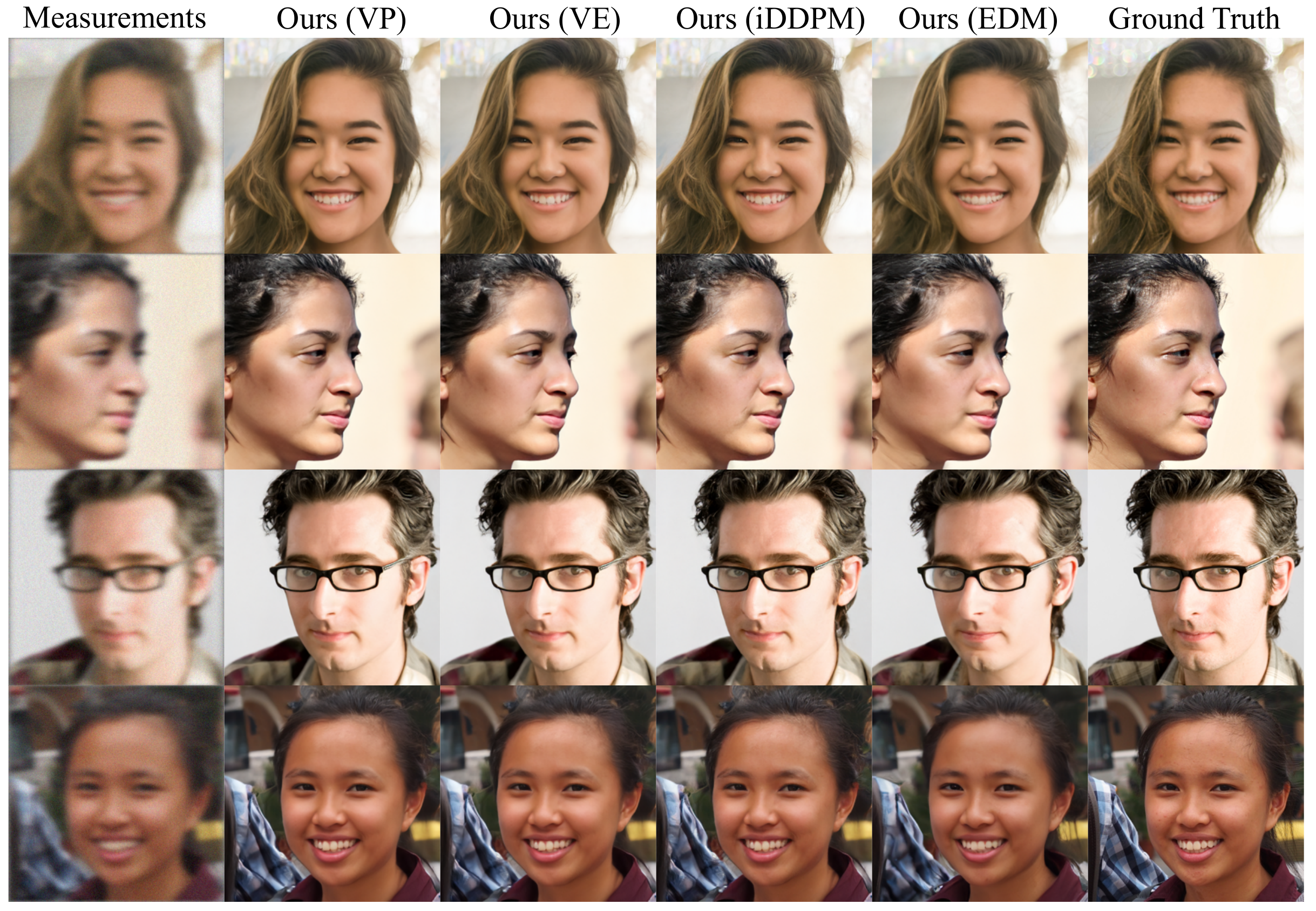}}
		\caption{Visual results of ansiotropic-Gaussian deblur}
			\label{fig:vagdeblur}
		\end{center}
\end{figure*}
\begin{figure*}[!ht]
	\begin{center}
		\centerline{\includegraphics[width=\columnwidth]{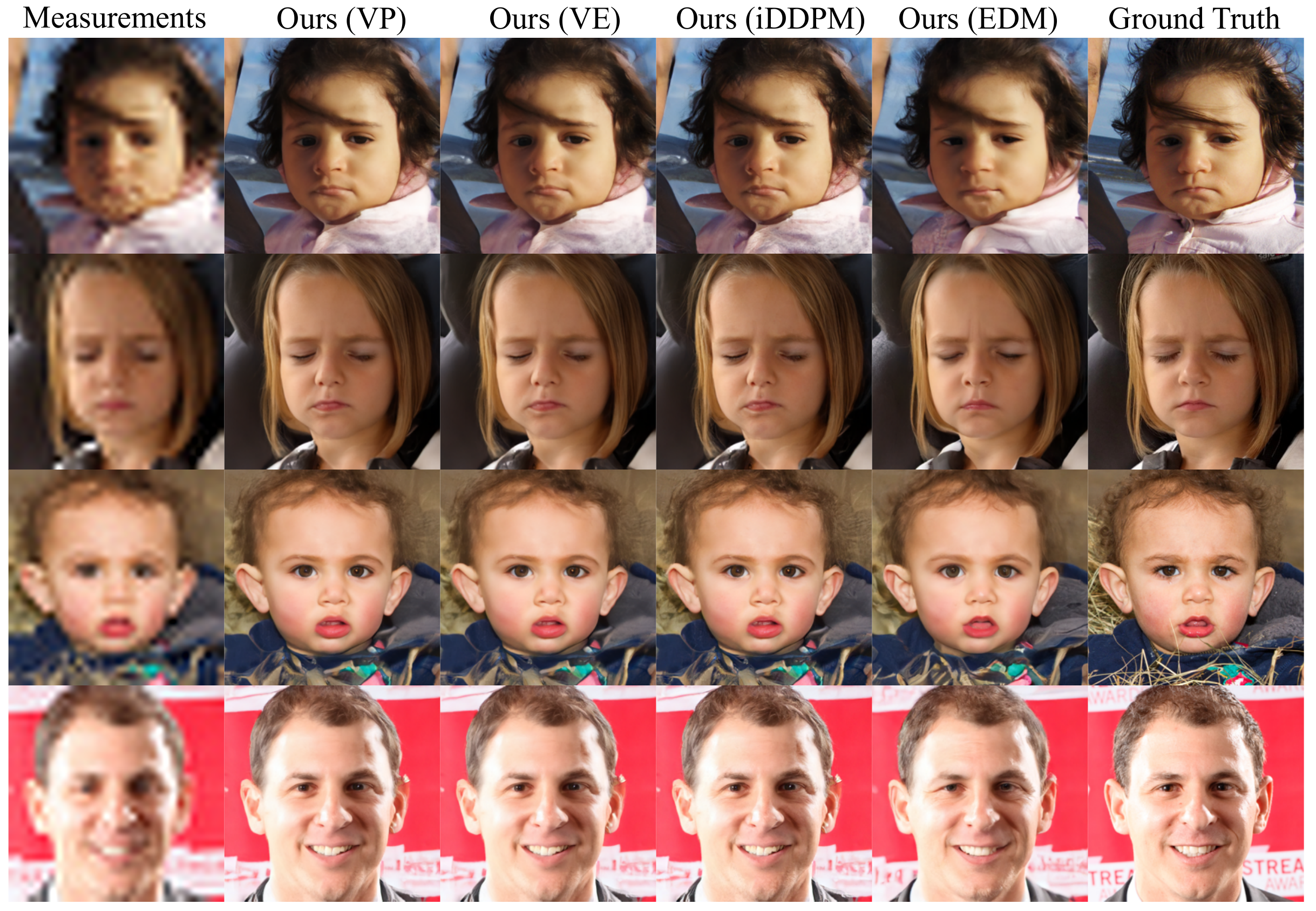}}
		\caption{Visual results of SR$\times8$}
			\label{fig:sr8}
		\end{center}
\end{figure*}
\begin{figure*}[!ht]
	\begin{center}
		\centerline{\includegraphics[width=\columnwidth]{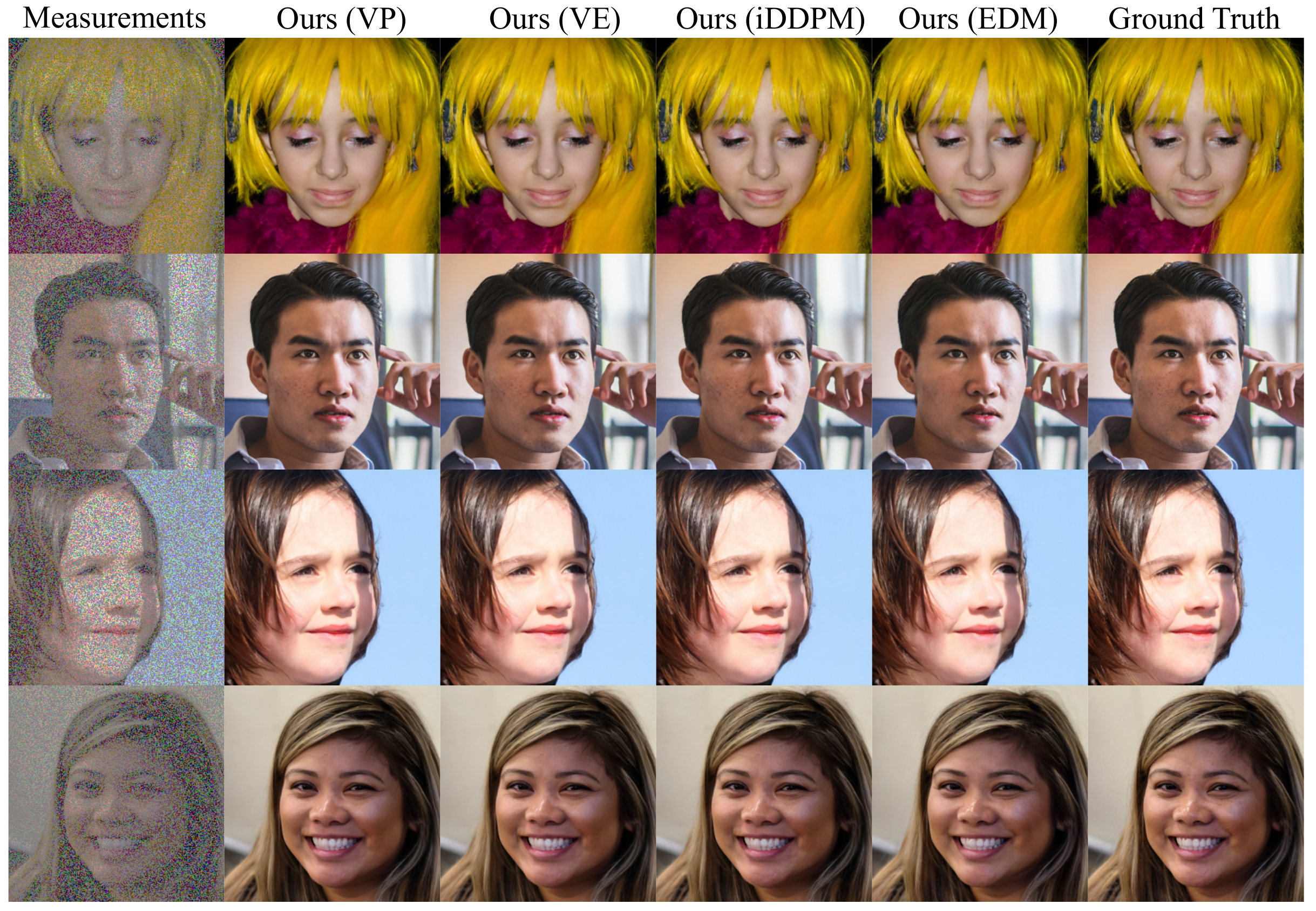}}
		\caption{Visual results of Inpainting with sampling ratio of 0.4-0.5}
			\label{fig:inpaint40}
		\end{center}
\end{figure*}

\begin{figure*}[!ht]
	\begin{center}
		\centerline{\includegraphics[width=\columnwidth]{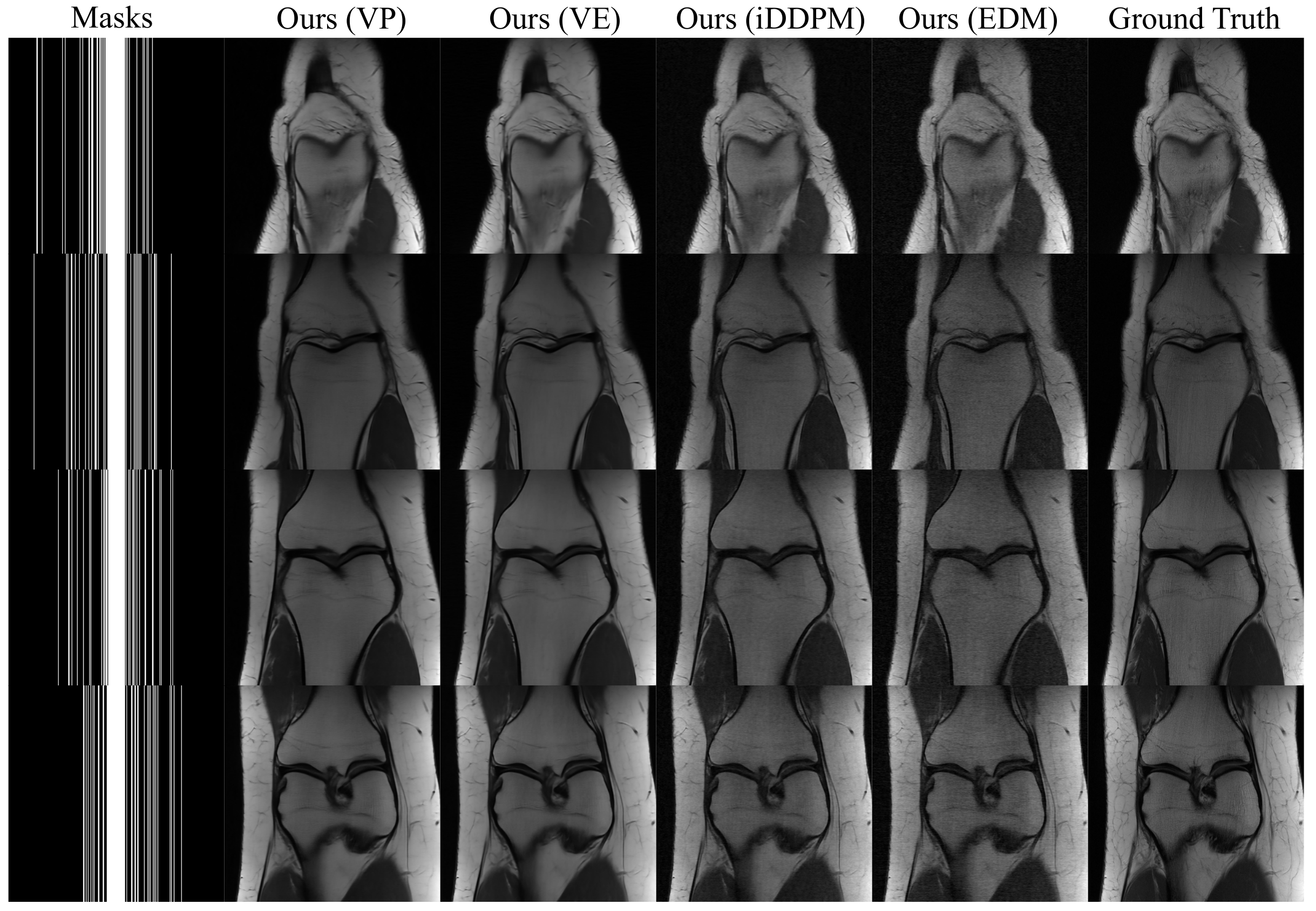}}
		\caption{Visual results of MRI with $8\times$ Gaussian 1D sampling}
			\label{fig:mrig1d8}
		\end{center}
\end{figure*}

\begin{figure*}[!ht]
	\begin{center}
		\centerline{\includegraphics[width=\columnwidth]{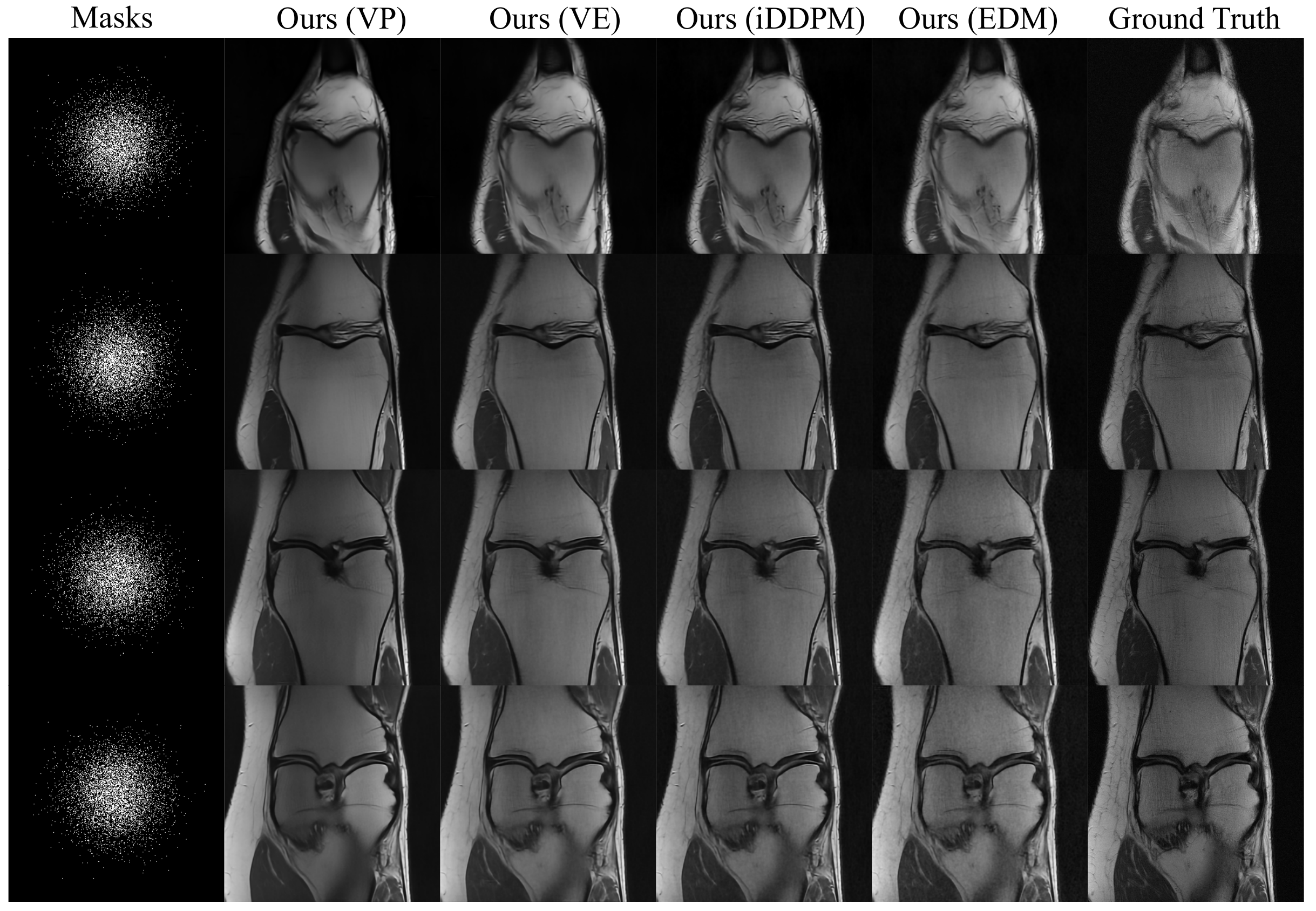}}
		\caption{Visual results of MRI with $15\times$ Gaussian 2D sampling}
			\label{fig:mrig2d15}
		\end{center}
\end{figure*}

\begin{figure*}[!ht]
	\begin{center}
		\centerline{\includegraphics[width=\columnwidth]{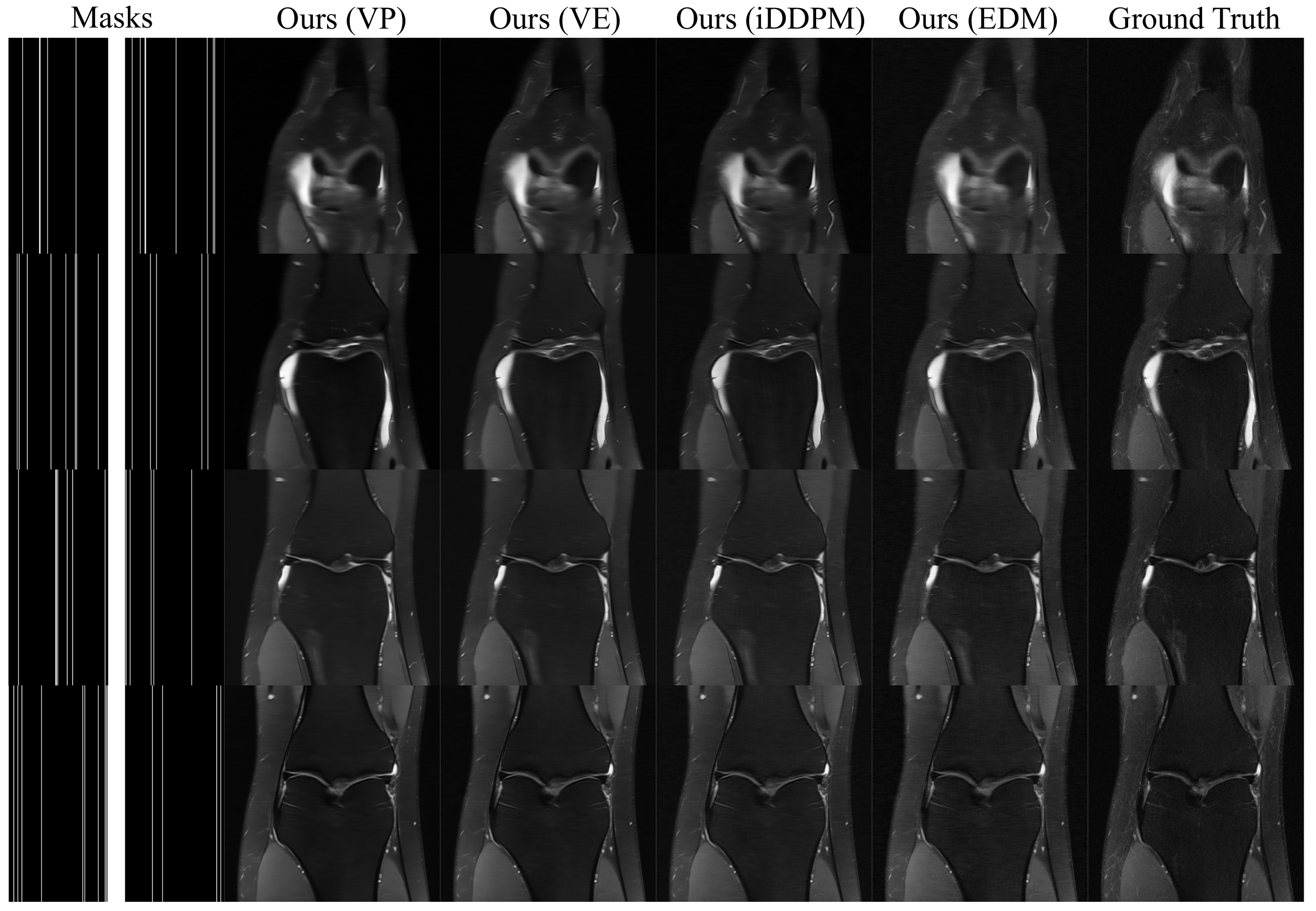}}
		\caption{Visual results of MRI with $8\times$ Uniform 1D sampling}
			\label{fig:mriu1d8}
		\end{center}
\end{figure*}